\documentclass{article}

\usepackage{arxiv}

\usepackage[utf8]{inputenc} 
\usepackage[T1]{fontenc}    
\usepackage{hyperref}       
\usepackage{url}            
\usepackage{booktabs}       
\usepackage{amsfonts}       
\usepackage{nicefrac}       
\usepackage{microtype}      
\usepackage{lipsum}
\usepackage{times}
\usepackage{epsfig}
\usepackage{graphicx}
\usepackage{amsmath}
\usepackage{amsthm}
\usepackage{amssymb}
\usepackage{bm}
\usepackage{nccmath}
\usepackage{multirow}
\usepackage{url}
\usepackage{microtype}
\usepackage{algorithm2e}
\usepackage{array}
\usepackage{relsize}
\usepackage{cite}
\usepackage{hyperref}
\usepackage{subcaption}
\usepackage{makecell}
\usepackage{color,soul}

\def \g{\mathbf{g}}

\def \r{\mathbf{r}}
\def \t{\mathbf{t}}
\def \h{\mathbf{h}}
\newcommand{\pwn}[1]{\|{#1}\|_{\cdot 2}}
\newcommand{\innerprod}[2]{\langle #1, #2\rangle}

\def \xh{\hat{x}}
\def \yh{\hat{\mathbf{y}}}
\def \zh {\hat{\mathbf{z}}}

\def \y {\mathbf{y}}
\def \z {\mathbf{z}}
\def \x {\mathbf{x}}

\newcommand{\genm}[3]{
	\left[\begin{array}{cc}
		#1  &  #3 \\
		#3 & #2 
	\end{array}
	\right]
}

\newtheorem{prop}{\bf Proposition}
\graphicspath{{Figures/}}

\title{Structurally  Adaptive  Multi-Derivative  Regularization
	for Image Recovery from Sparse Fourier Samples}

\author{  
		Sanjay Viswanath, \;\; Manu Ghulyani, and  \;\; 
Muthuvel Arigovindan\thanks{Corresponding author (mvel@iisc.ac.in)} \\
		Imaging Systems Lab\\  
		Department of Electrical Engineering\\
		Indian Institute of Science (IISc)\\ Bangalore 560012, India\\
				}

\begin{document}
	\maketitle
	
	\begin{abstract}
		The importance  of regularization   has been well  established  in 
		image 
		reconstruction---which is the computational inversion of imaging 
		forward 
		model--with applications including deconvolution for microscopy,  
		tomographic
		reconstruction,  magnetic
		resonance imaging, and so on.  Originally, the primary role of the 
		regularization was to 
		stabilize  the computational inversion of the imaging forward model 
		against 
		noise.  However, 
		a recent  framework pioneered by Donoho and others, known as 
		compressive 
		sensing, brought 
		the role of regularization beyond the stabilization of inversion.  It 
		established a possibility
		that regularization can recover full images 
		from  
		highly 
		undersampled measurements.
		However, it was observed that the quality of reconstruction 
		yielded by compressive
		sensing methods falls abruptly when the under-sampling and/or  
		measurement 
		noise goes
		beyond a certain threshold.  Recently developed learning-based methods  
		are 
		believed to
		outperform the compressive sensing methods without a steep drop in 
		the 
		reconstruction 
		quality  under such imaging conditions.  However,
		the need for training data limits their applicability.  In this paper, 
		we 
		develop a regularization
		method that outperforms compressive sensing methods as well as 
		selected  
		learning-based methods,
		without any need for training data.  The regularization is constructed 
		as a spatially varying
		weighted sum  of   first- and  canonical second-order derivatives, with 
		the 
		weights determined
		to be adaptive to the image structure;  the weights are determined  
		such that  the attenuation of sharp 
		image 
		features---which is inevitable with the use of any regularization---is 
		significantly reduced.
		We demonstrate the effectiveness of the proposed method by performing 
		reconstruction
		on sparse Fourier samples simulated from a variety of MRI images.
	\end{abstract}

	\keywords{Adaptive Regularization, Total Variation, Magnetic Resonance 
	Imaging, 
		Image 
		Reconstruction}
	
	\section{Introduction}
	
	Magnetic resonance imaging (MRI) is one of the important medical 
	imaging modalities  in modern medicine for both clinical practice and 
		medical 
		research.  Its importance stems
	from the fact that it is absolutely non-invasive.  An MRI system measures 
	samples of the Fourier transform of a
	cross-sectional  plane of interest \cite{MRI_Haacke1999magnetic, 
		MRI_Pipe1999sampling}.  The measured samples in the Fourier plane 
		typically lie on 
	a trajectory 
	that is determined by the physical constraints of the imaging system.   
	In most systems,  the measured samples are transformed  into a set of
	samples on a cartesian grid via interpolation, such that the required real 
	space 
	image is obtained by simple
	inverse discrete  Fourier transform \cite{MRI_Haacke1999magnetic, 
		MRI_Pipe1999sampling}.   However, this type of image reconstruction 
		method 
	leads to serious artifacts,
	when the measured  Fourier samples are not sufficiently dense.  A dense 
	sampling
	is impractical in cases 
	where it is required to reduce the scan time 
	\cite{MRI_Sparse_Lustig2007sparse}.
	In such cases,  the method of compressive sensing,  pioneered by Donoho and 
	others \cite{CS_donoho2006compressed,CS_candes2006compressive, 
		CS_candes2006near, 
		CS_tsaig2006extensions, CS_candes2006quantitative, 
		CS_baraniuk2007compressive, 
		CS_elad2007optimized, CS_peyre2010best,
		CS_tropp2010computational, CS_duarte2011structured}, 
	becomes  a better 
	alternative 
	\cite{ 
		CS_candes2006robust, MRI_Sparse_Lustig2007sparse, 
		CS_MRI_trzasko2008highly, 
		CS_MRI_ma2008efficient, 
		MRI_CS_lustig2008compressed, MRI_Bresler, CS_MRI_patel2011gradient, 
		CS_MRI_tamada2014two}, as it can reconstruct good quality images from
	sparse samples.   Unlike direct inversion, the method of compressive 
	sensing formulates the reconstruction problem as an inverse problem and 
	solves
	it by employing  numerical optimization as given in 
	Figure \ref{fig:basic}.  In this scheme,  a cost function 
	composed of the following parts is 
	constructed:   
	\begin{itemize}
		\setlength\itemsep{-0.2em}
		\vspace{-0.05cm}
		\item
		A data error term that measures the goodness of fit of the candidate 
		image 
		to the measured
		data through the imaging forward model, which can be expressed as
		\begin{equation}
			\label{eq:dataerr}
			F(s,{\bf m}, T) = \| Ts-{\bf m} \|_2^2,
		\end{equation}
		where $T$ represents a  sampling operator in the Fourier domain that 
		returns a  vector of 
		complex Fourier samples from
		the candidate image $s$,  ${\bf m}$ denotes the vector of measured 
		samples,    and 
		$\|\cdot\|_2$ represents the
		magnitude of its vector/matrix argument
		which is the  square root of the sum of squares of  its elements.
		\item
		An image  regularization  cost $R(s)$,  that measures the overall 
		roughness 
		of the 
		candidate image  through  image derivatives.  The most commonly used 
		forms   for
		$R(s)$ are given below \cite{TV1_Rudin, Scherzer_tv2_98, HS_13}:
		\begin{align}
			\label{eq:regtv1}
			R_1(s) & = \sum_{\bf r} \| \nabla s(\r) \|_2, \\
			\label{eq:regtv2}
			R_2(s) & = \sum_\r \| \nabla^2 s(\r) \|_F,  \\
			\label{eq:reghs}
			R_{hs}(s) & = \sum_\r \| \nabla^2 s(\r) \|_{{\cal S}(1)},
		\end{align}
		where  $\nabla$  and $\nabla^2$ are gradient and Hessian operators.  
		The 
		symbols 
		$\| \cdot \|_2$  and  $\| \cdot \|_F$ denotes vector and matrix norms 
		that 
		are
		computed as the square root of the sum of squares of the elements.
		The symbol $\| \cdot \|_{{\cal S}(1)}$
		denotes the Nuclear-norm, which is the sum of  modulus of Eigenvalues 
		of 
		its matrix argument.
		$R_{hs}(s)$ has proven to be one of the best performing 
		regularizations. 
	\end{itemize}
	The numerical  optimization program evaluates  a  weighted sum of the form
	$J(s,{\bf m}, T, \lambda) = F(s,{\bf m},T) + \lambda R(s)$  where $\lambda$ 
	is 
	the user-defined
	weight determining the trade-off between the data error and the image 
	roughness.
	This cost is evaluated along with its gradient for a series of candidates 
	images  
	(varying $s(\r)$), 
	and determines the minimum, 
	which becomes the required reconstructed  image   $s^*(\r)$.  We will denote
	this operation by  
	\begin{align}
		\label{eq:recsymbna}
		s^{*} =   \mathlarger{{\cal M}}\Big(
		\scalebox{.7}[1.0]{\bf Data}=\{T, {\bf m}\},  \;\;
		\lambda \Big).
	\end{align}
	In the figure,  
	the symbol,  $\hat{s}(\r)$  denotes
	the initialization required for the numerical optimization. Regularized 
	reconstruction methods that use
	the forms of regularization given in the equations \eqref{eq:regtv1}, 
	\eqref{eq:regtv2}, and \eqref{eq:reghs}
	are also called total variation (TV) methods. We will refer to this type of 
	methods
	as the derivative-based regularization methods.
	Further  developments in derivative-based regularization include total
	generalized variation (TGV) \cite{TGV2},  which is believed to give 
	superior 
	reconstruction.
	A detailed overview of such methods has been provided in the 
	supplementary 
	material.
	It has been observed 
	that compressive sensing
	methods also give reconstruction  artifacts, when the sampling is very 
	sparse 
	and/or noise is very high.
	In  recent years,  there   has been increased focus on developing deep 
	learning based methods \cite{Deep_CT_Jin2017deep, 
	MRI_Deep_schlemper2017deep, 
		MRI_Deep_jin2017deep, Deep_MRI_1, DAGAN, Deep_MRI_3,   
		MRI_Deep_quan2018compressed, Deep_MRI_zhang2020deep},  since it is 
		generally 
	believed that such 
	methods do not produce artifacts in the reconstruction. 
	This belief is based on the following fact:
	while  regularization methods impose  ad hoc prior models on the image 
	characteristics,  
	deep-learning methods encompass  natural prior information derived from 
	training data. 
	However, deep 
	learning-based methods require
	training data in the form of  thousands of image pairs, where each pair 
	contains a typical measured image and the 
	corresponding ground-truth image that might have generated the 
	measurement.  
	Such a requirement limits the 
	applicability of deep learning based methods.  Furthermore, stability 
	issues 
	have also been reported in such deep learning methods for image 
	reconstruction 
	\cite{Deep_Instability}.
	\begin{figure}[ht]
		\centering
		\includegraphics[width= 0.6\textwidth]{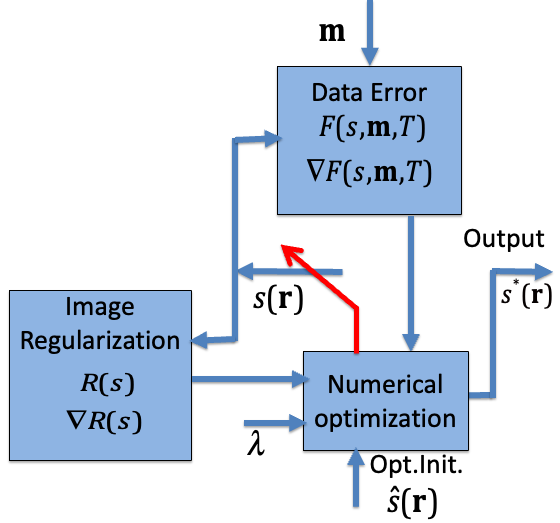}
		\caption{Schematic of regularized image reconstruction.  ${\bf m}$ 
		denotes 
			the measured data.
			$F(s,{\bf m},T)$ represents the data error with $T$ denoting the 
			measurement model of the imaging
			system.  $R(s)$ is the measure of image roughness 
			(regularization).  
			$\nabla F(s,{\bf m},T)$   and 
			$\nabla R(s)$ represent gradients of  $ F(s,{\bf m},T)$   and 
			$R(s)$ respectively, which are required for the numerical 
			optimization. 
			$\lambda$ is the user parameter that determines the trade-off 
			between data error and
			image roughness.
		}
		\label{fig:basic}
	\end{figure}
	
	Our goal here is to develop a novel type of reconstruction method that has 
	the desirable
	traits of both deep-learning and  derivative-based  regularization 
	methods.  In other
	words, our interest is to
	develop a method that neither results in a significant amount of 
	distortions under the conditions
	of severe noise and under-sampling, nor requires any training  
	data.    To 
	brief the concepts
	involved in our approach,  
	we first need to understand why derivative-based regularization methods 
	give 
	distortions in the reconstructions when noise or under-sampling  is severe.
	To this end, we first note that employing derivatives in regularization has 
	the
	following desirable effects:  (i) minimizing derivative magnitudes summed
	over pixels  leads to  sparse distribution of   large derivatives in the 
	reconstruction, which results in the  elimination of noise;
	(ii)  additionally, in the case of MRI reconstruction,  minimizing 
	derivative magnitudes 
	counter-acts with  oscillations caused 
	by under-sampling  in  Fourier space, thereby leading to the interpolation 
	of 
	lost 
	Fourier samples. However, minimizing derivative has also an undesirable 
	effect:
	as high-resolution image features will have high derivative values,  using 
	derivative-based regularization inevitably leads
	to distortions that cause  some loss of resolution.     Our goal is to 
	build an improved  
	derivative-based regularization scheme  with the following property:   it 
	should have a reduced penalizing effect
	on the derivatives that are part of the image structure without losing its 
	penalizing
	effect on the derivatives originating from noise and oscillations.  We call 
	such property
	structural adaptivity.  Such a regularization method with structural 
	adaptivity 
	will hopefully have the desirable traits:  it will not give 
	significant 
	amount of
	distortion under severe noise and under-sampling, and it will not require 
	training data.

	To build a regularization  scheme with structural adaptivity,  we draw 
	inspiration from
	the human visual system.    It is a well-known fact that  the  human  
	visual system can
	localize  object boundaries in the presence of a very large amount of noise 
	because of  its
	inherent  ability  to process images simultaneously at multiple resolutions 
	\cite{MR_tanimoto1975hierarchical,  MR_ghassemian2001retina}.  Because of 
	this
	ability,   the human 
	visual system can   discriminate sharp jumps in the intensity that is part 
	of an object boundary
	against sharp jumps caused by noise.    Inspired by the ability of the 
	human visual system,  we build our   structural adaptive regularization   
	in the form of 
	a multiresolution scheme.  We do this by   constructing 
	a  series of reconstruction modules with progressively increasing target 
	resolutions,
	where the reconstruction from a module with lower target resolution serves 
	as a 
	guide
	for building  the adaptive regularization for the module with a higher 
	target
	resolution.    This leads to two questions: 
	how  to define  the adaptive regularization, given a guide image, and 
	how to  define modules with progressively increasing target resolution.
	The first question is answered in Section \ref{sec:arr}, and the second
	question is answered in Section \ref{sec:mr}.  These sections only deal with
	the main concepts involved in our method.  The optimization tools necessary 
	for
	the practicable realization of our method are developed in the 
	supplementary Sections
	\ref{sec:adreg} and  \ref{sec:wtcomp},  and these also constitute  our 
	important
	technical contributions.

	In Section \ref{sec:exp}  we present an extensive set of reconstruction 
	results  and
	demonstrate that the proposed method outperforms
	all derivative-based regularization methods,  and prominent deep learning
	methods. The key factor that makes  our methods superior  to deep learning 
	methods 
	and other regularization
	methods is that the adaptive regularization  does not penalize  high images 
	derivatives that are
	part of the image structures;   this is in contrast with standard 
	regularization methods, which are unable to discriminate
	derivatives that are part of the structures against derivatives 
		originating
	from noise. This is made
		possible by the structural adaptivity resulting from the 
		multiresolution 
	scheme that we propose.
	The idea of  multi-resolution
	has been extensively used for other processing tasks such as image 
	registration 
	\cite{MR_Huttenlocher_registration, MR_alhichri2002multi_registration, 
		MR_bunting2010registration},   and image
	segmentation \cite{MR_kim2003fast_segmentation, 
	MR_tang2006multi_segmentation, 
		MR_janoos2007multi_segmentation} with strong inspiration derived from 
		the
	way the human visual system operates      
	\cite{MR_tanimoto1975hierarchical,  MR_ghassemian2001retina}.  
	However,  the power of multiresolution
	has not been used for gaining robustness in image reconstruction so far,  
	and
	the present work
	uses the power of multiresolution  for the first time.
	The method that we develop here is an extension of the method we  recently 
	published 
	called
	Combined Order Regularization with Optimal Spatial Adaptation (COROSA) 
	\cite{COROSA}. 
	
	\section{The basic adaptive regularized reconstruction}
	\label{sec:arr}
	
	As explained before, our  multiresolution-based adaptive regularization 
	scheme will be in the form of a series
	of reconstruction modules;  each module performing an adaptive 
	regularized
	reconstruction by using the information from the guide image, which is 
	obtained as the reconstruction
	from the previous module. Obviously, the key to the success of this scheme 
	is
	the construction of a basic adaptive regularized reconstruction  module 
	that 
	can extract information for 
	structural adaptation from  the guide image.  To this end, we make use of 
	the
	following observation:   the relative spatial distribution of  the image 
	gradient and  the
	Eigenvalues of image Hessian has important information about the image 
	structure.   We propose to 
	exploit this towards constructing an adaptive regularization  that 
	has  
	the ability 
	to discriminate
	derivatives that are part of the structures against the derivatives that 
	originate 
	from noise (structural adaptivity). 
	The regularization is  constructed as a spatially varying weighted sum of 
	the above-mentioned three terms.
	The weights are determined adaptively  using the guide image
	such that the resulting regularization has the desired structural adaptivity.
	We do this in this section in three stages. 
	In Section \ref{sec:arrcomp},  we describe the basic building blocks to be 
	used
	in adaptive regularization.  Next,  we develop the adaptive regularized
	reconstruction method using these building blocks in 
	Section\ref{sec:arrdef},  and
	in Section \ref{sec:guide2wt} we develop a method to extract the 
	information for 
	adaptation from a given guide image.

	\subsection{Gradient and  canonical second derivatives}
	\label{sec:arrcomp}
	
	For a continuous domain function  $p({\bf r})$,  its gradient denoted by 
	$\nabla 
	p({\bf r})$ 
	is a vector function of the form 
	$\nabla p({\bf r}) = \left[\frac{\partial}{\partial x}p({\bf r})\;\; 
	\frac{\partial}{\partial y}p({\bf r})  \right]^T$.
	We use the same  notation to denote the discrete  gradient  of a discrete 
	image   $s({\bf r})$,
	i.e., we write $\nabla s({\bf r}) = \left[(d_x*s)(\r)\;\;(d_y*s)(\r) 
	\right]^T$, where $*$ denotes convolution,
	and,  $d_{x}(\r)$  and $d_{y}(\r)$ are the filters implementing discrete 
	equivalents
	of the first order-derivatives, 
	$\frac{\partial}{\partial x}$, and   $\frac{\partial}{\partial y}$ 
	respectively.  
	Summing this gradient magnitude over all pixels of the image $s(\r)$ gives 
	the first-order total variation regularization    of equation 
	\eqref{eq:regtv1}.

	\noindent
	For a continuous domain function,  $p({\bf r})$, the elementary second 
	derivative 
	operators are given by
	${\cal D}_{xx}p(\r)    =  \frac{\partial^2}{\partial x^2}p({\bf r})$, 
	${\cal D}_{yy}p(\r)    =  \frac{\partial^2}{\partial y^2}p({\bf r})$,   and 
	${\cal D}_{xy}p(\r)    =   \frac{\partial^2}{\partial x\partial y}p({\bf 
		r})$.   The Hessian is given by
	$\nabla^2 p(\r) = \genm{{\cal D}_{xx}p(\r)}{{\cal D}_{yy}p(\r)}{{\cal 
			D}_{xy}p(\r)}$. The Eigenvalues
	of Hessian are given by
	\begin{align}
		{\cal H}^+p({\bf r}) & = \frac{1}{2}\left({\cal L}p({\bf r})+{\cal 
		C}p({\bf 
			r}) \right), \\
		{\cal H}^-p({\bf r}) & = \frac{1}{2}\left({\cal L}p({\bf r})-{\cal 
		C}p({\bf 
			r})  \right),
	\end{align}
	where
	\begin{equation}
		\label{eq:descdef}
		{\cal C}p(\r) = 
		\sqrt{\left({\cal D}_{xx}p(\r)   -  {\cal D}_{yy}p(\r) \right)^2  
			+   4{\cal D}_{xy}p(\r) },
	\end{equation}
	and  ${\cal L}p({\bf r})$ is the well-known Laplacian operator given by
	${\cal L}p({\bf r}) =  {\cal D}_{xx}p(\r)   +  {\cal D}_{yy}p(\r)$. 
	For brevity, we call these Eigenvalues,  the canonical second derivatives 
	(CSDs). The term
	``canonical"  is justified because of the    following reasons:  
	(i) these quantities
	themselves are second derivative taken along the directions that make angles
	$\theta(\r)$ and $\theta(\r)+\pi/2$ with the $x$-axis,   where $\theta(\r)$ 
	is 
	the $\r$-dependent angle
	determined by the elementary derivatives 
	${\cal D}_{xx}p(\r)$, ${\cal D}_{yy}p(\r)$, and  ${\cal D}_{xy}p(\r)$;
	(ii) any  second derivative taken along a direction that makes  an 
	arbitrary 
	angle, $\vartheta$, can be expressed
	as  $\mu {\cal H}^+p({\bf r}) + (1-\mu) {\cal H}^-p({\bf r})$  where $\mu$ 
	is a 
	real number in the range 
	$[0,1]$ that is determined by the angle $\vartheta$; (iii) these quantities 
	are 
	invariant with respect to
	rotations of the image.  
	For a discrete image, $s(\r)$,  we use the same notations to denote the 
	CSDs, 
	with an understanding that
	the constituent operators ${\cal D}_{xx}s(\r)$, ${\cal D}_{yy}s(\r)$, and  
	${\cal D}_{xy}s(\r)$ are implemented
	by means of discrete convolutions with appropriately designed filters 
	denoted 
	by 
	$d_{xx}(\r), d_{yy}(\r)$, and $d_{xy}(\r)$. In the supplementary material,  
	we 
	have given the implementation
	of filters    $d_x$,  $d_y$, $d_{xx}$,   $d_{yy}$, and $d_{xy}$.

	\subsection{Adaptive regularized reconstruction}
	\label{sec:arrdef}
	
	As mentioned before, the most widely used regularizations that are 
	considered 
	to be well-performing 
	are the ones expressed in  Equations (\ref{eq:regtv1}), (\ref{eq:regtv2}), 
	(\ref{eq:reghs}).  $R_1(s)$
	is constructed by simply summing $\|\nabla s({\bf r})\|_2$ across all 
	pixels. 
	Further,  $R_{hs}(s)$
	is constructed by summing   $|{\cal H}^+s({\bf r})|+|{\cal H}^-s({\bf 
			r})|$   
	across  all pixels, whereas,
	$R_2(s)$ is constructed by summing  $\sqrt{({\cal 
	D}_{xx}s(\r))^2+({\cal 
				D}_{yy}s(\r))^2+2({\cal D}_{xy}s(\r))^2}$ across  all pixels. 
				It turns
	out that this is the same as summing $\sqrt{({\cal H}^+s({\bf 
	r}))^2+({\cal 
				H}^+s({\bf r}))^2}$ across
	all pixels.  In our viewpoint, although  these regularizations efficiently 
	remove noise,  they have
	the  disadvantage that they do not exploit any  pattern in the relative 
	distribution of these quantities.
	We hypothesize  that the quantities  $\|\nabla s({\bf r})\|_2$,  ${\cal 
		H}^+s({\bf r})$, and ${\cal H}^-s({\bf r})$, 
	will have some pattern in their spatial distribution,  and this can be 
	exploited to reduce  the loss of  resolution
	that is normally caused by derivative-based regularizations.  To this end,  
	we 
	construct the following 
	regularization:
	\begin{equation}
		\label{eq:regsa}
		R(s,\hat{\beta}_{1},  \hat{\beta}_{2}, 
		\hat{\beta}_{3})   = \sum_{\r} \left(  
		\hat{\beta}_1(\r)\|\nabla s({\bf r})\|_2  +   \hat{\beta}_2(\r) |{\cal 
			H}^+s({\bf r})|   
		+  \hat{\beta}_3(\r) |{\cal H}^-s({\bf r})|
		\right) 
	\end{equation}
	where $s(\r)$ is the candidate for reconstructed image,  and 
	$\hat{\beta}_1(\r)$,  $\hat{\beta}_2(\r)$, and   $\hat{\beta}_3(\r)$  are 
	adaptive weight images satisfying 
	$\sum_i\hat{\beta}_i(\r)=1$ and $\hat{\beta}_i(\r) \ge 0, i=1,2,3$.   The 
	idea 
	is to choose the weights $\hat{\beta}_1(\r)$, 
	$\hat{\beta}_2(\r)$, and   $\hat{\beta}_3(\r)$  appropriately to exploit 
	the  
	pattern in the distribution of
	$\|\nabla s({\bf r})\|_2$,  ${\cal H}^+s({\bf r})$, and ${\cal H}^+s({\bf 
	r})$. 
	At the same time,  because of the
	constraint, $\sum_i\hat{\beta}_i(\r)=1$, the penalizing effect on the noise 
	will 
	not be compromised.  
	The proposed image reconstruction using this regularization  is  
	demonstrated 
	in the schematic of Figure 
	\ref{fig:adreg}.  With given measured samples, ${\bf m}$,   adaptive 
	weight  
	images $\hat{\beta}_1(\r)$,
	$\hat{\beta}_2(\r)$, and   $\hat{\beta}_3(\r)$, and  regularization weight 
	$\lambda$, the numerical optimization
	module evaluates the image regularization  $R(s,\hat{\beta}_{1},  
	\hat{\beta}_{2}, \hat{\beta}_{2})$ and data error
	$F(s,{\bf m},T)$ along with their gradients for a series of candidate 
	images 
	(optimization variable $s(\r)$)
	and determines the empirical minimum of the sum   $J(s,T,{\bf m},\lambda)=
	F(s,T,{\bf m}) + \lambda R(s,\hat{\beta}_{1},  \hat{\beta}_{2}, 
	\hat{\beta}_{2})$.    This empirical minimum,  $s^*(\r)$,  is the  
	candidate  
	$s(\r)$   that has a modulus of
	the gradient of the cost $J(s,T,{\bf m},\lambda)$  lower than certain 
	user-defined tolerance.
	\begin{figure}[htbp]
		\centering
		\includegraphics[width= 0.6\textwidth]{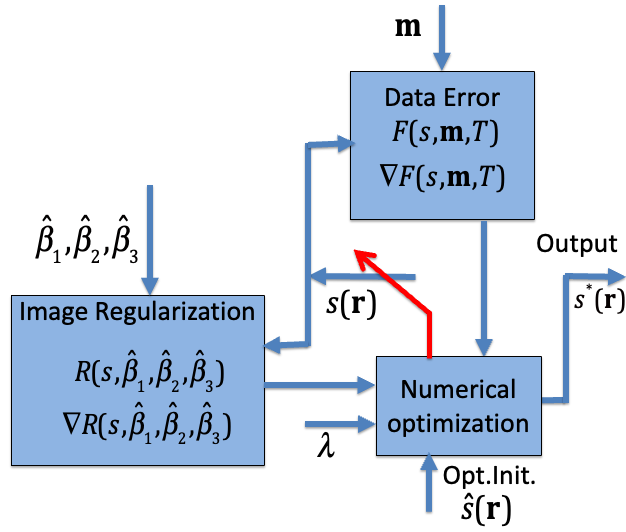}
		\caption{Schematic of adaptive regularized image reconstruction from 
		given 
			adaptive weights}
		\label{fig:adreg}
	\end{figure}
	In our proposed scheme,   the required adaptive weights are determined 
	from  a 
	given guide  image $\hat{s}(\r)$,  that is supplied by the multi-resolution 
	chain 
	to be elaborated in Section \ref{sec:mr}.
	We will deal with the issue of finding the  adaptive weights  
	$\hat{\beta}_{1}$,  $\hat{\beta}_{2}$, 
	and $\hat{\beta}_{3}$  in Section \ref{sec:guide2wt}. 
	A point to be noted is that we will use $\hat{s}(\r)$  also for 
	initializing 
	the numerical optimization.  It should be emphasized that 
	currently available generic optimization  tools  are unsuitable for this  
	optimization problem and we develop  novel optimization tools for this 
	purpose  in the   supplementary material. 
	We symbolically  denote this  numerical  optimization based reconstruction 
	operation by
	\begin{align}
		\label{eq:recph1}
		s^* =   \mathlarger{{\cal M}_{ad}}\Big(
		\scalebox{.7}[1.0]{\bf Init}
		=\hat{s},  \;
		\scalebox{.7}[1.0]{\bf AdapWt} = \{\hat{\beta}_{1},  \hat{\beta}_{2}, 
		\hat{\beta}_{2}\}, 
		\scalebox{.7}[1.0]{\bf Data}=\{T, {\bf m}\},  \;\;
		\lambda \Big). \nonumber
	\end{align}


	\subsection{Determining the adaptive weights from guide image}
	\label{sec:guide2wt}
	
	We next address the problem  of determining the adaptive weights,  
	$\hat{\beta}_1(\r)$,  $\hat{\beta}_2(\r)$,
	and   $\hat{\beta}_3(\r)$ from the given guide image $\hat{s}(\r)$.
	We propose to determine these from
		$\hat{s}(\r)$ via minimizing the same regularization cost 
	$R(s,\beta_1,\beta_2,\beta_3)$
	w.r.t.  $\beta_1(\r)$,  $\beta_2(\r)$,  and   $\beta_3(\r)$ with the 
	substitution 
	$s(\r)=\hat{s}(\r)$.
	We also impose the constraint that $\beta_1(\r)+\beta_2(\r)+\beta_3(\r)=1$. 
	The idea behind this approach is that  the quantities $\|\nabla {s}({\bf 
		r})\|_2$, 
	$ {\cal H}^+s({\bf r})$ and ${\cal H}^-s({\bf r})$  evaluated with 
	substitution 
	$s(\r)=\hat{s}(\r)$ will have a pattern that is similar to that of  the 
	required reconstruction provided that  $\hat{s}(\r)$  is a good 
	approximation
	to the required reconstruction.
	Hence, the resulting  adaptive regularization $R(s,  \hat{\beta}_1, 
	\hat{\beta}_2,  \hat{\beta}_3)$
	to be  used in the reconstruction scheme of  Figure \ref{fig:adreg} will 
	have 
	a lesser penalizing
	effect on the derivatives that confers to the pattern exhibited by 
	guide 
		$\hat{s}(\r)$.  At the same time, 
	since the constraints,  $\sum_i\beta_i(\r)=1$ and $\beta_i(\r) \ge 0, 
	i=1,2,3$ 
	are imposed,
	the adaptive regularization $R(s,  \hat{\beta}_1, \hat{\beta}_2,  
	\hat{\beta}_3)$  will not have any
	reduction in the penalizing effect on noise as well as other related 
	artifacts.  One technical
	point is that  it is not well-defined to minimize 
	$R(s,\beta_1,\beta_2,\beta_3)$  
	w.r.t.     $\beta_1(\r)$,  $\beta_2(\r)$,  and   $\beta_3(\r)$
	because $R(s,\beta_1,\beta_2,\beta_3)$   is linear on   $\beta_1(\r)$, 
	$\beta_2(\r)$,  and   
	$\beta_3(\r)$.
	We instead minimize   $R(s,\beta_1,\beta_2,\beta_3)+\tau L(\beta_1,\beta_2,
	\beta_3)$ where $L(\beta_1,\beta_2, \beta_3)=-\sum_{\r} \log 
	(\beta_1(\r)\beta_2(\r)\beta_2(\r))$.
	The relevant mathematical explanation for this augmentation  with
	$L(\beta_1,\beta_2, \beta_3)$
	is given in the supplementary material.   Further,  the minimization 
	algorithm 
	is developed
	in the supplementary material. We denote this minimization operation by 
	\begin{equation}
		\label{eq:recph2}
		(\hat{\beta}_1,  \hat{\beta}_2, \hat{\beta}_3) =\mathlarger{{\cal 
		M}_w}\Big(
		\scalebox{.7}[1.0]{\bf AdapGd}=\hat{s},   \tau \Big).
	\end{equation}
	Given the guide image,  $\hat{s}$, process of getting the adaptive 
	weights,  
	$\hat{\beta}_1(\r)$,  $\hat{\beta}_2(\r)$,  and $\hat{\beta}_3(\r))$ and 
	then 
	the process of obtaining
	the adaptively regularized reconstruction using 
	$\hat{\beta}_1(\r)$,  $\hat{\beta}_2(\r)$,  and $\hat{\beta}_3(\r))$
	can be represented as a cascade  as given below:
	\begin{align}
		\label{eq:recph12}
		s^* =   \mathlarger{{\cal M}_{ad}}\Big(
		\scalebox{.7}[1.0]{\bf Init}
		=\hat{s},  
		\scalebox{.7}[1.0]{\bf AdapWt} = 
		\mathlarger{{\cal M}_w}\Big(
		\scalebox{.7}[1.0]{\bf AdapGd}=\hat{s},   \tau \Big), 
		\;\;  
		\scalebox{.7}[1.0]{\bf Data}=\{T, {\bf m}\},  \;\;
		\lambda \Big).
	\end{align}
	This cascade  is represented in  Figure \ref{fig:adregg}. This cascade
	is the basic adaptive regularized  reconstruction module to be used for
	building the multiresolution method.
	It should be emphasized that, as an adaptation guide,  the quality of 
	$\hat{s}$  plays a
	crucial role in determining  the quality of the reconstruction $s^*$.  On 
	the 
	other hand,
	as an initializer for optimization,  the quality of $\hat{s}$  determines 
	the speed 
	of completion
	of the numerical optimization.
	\begin{figure}[htbp]
		\centering
		\includegraphics[width= 0.6\textwidth]{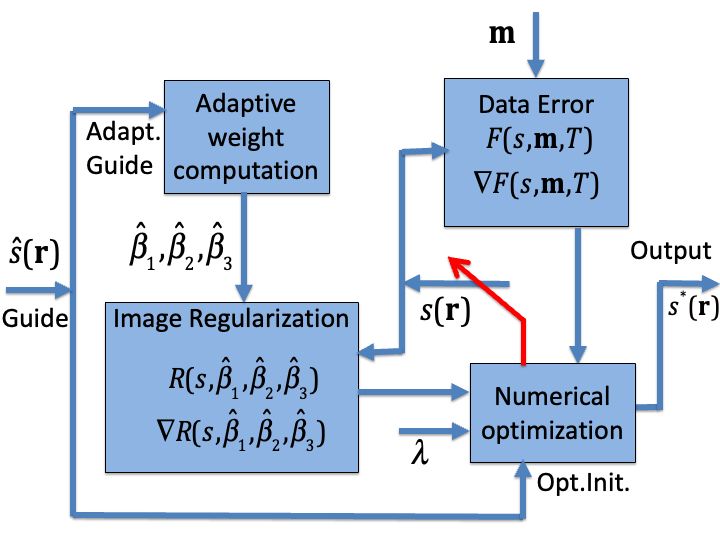}
		\caption{Schematic of adaptive regularized  reconstruction from given 
		guide 
			image}
		\label{fig:adregg}
	\end{figure}
	
	\section{Multi-resolution  based adaptive regularized  reconstruction 
	method}
	\label{sec:mr}
	
	Here we build the  multiresolution-based adaptive regularized 
	reconstruction method with a
	progressively increasing resolution by stacking  the basic  
	adaptive regularized  reconstruction developed in Section \ref{sec:arr}  
	(Figure \ref{fig:adregg}).
	To this end,  we introduce the
	notion of scale-$(j)$ reconstruction.  
	Recall that, in TV regularized reconstruction of Figure \ref{fig:basic}  and
	adaptive regularized reconstruction of Figure \ref{fig:adregg},   the 
	numerical optimization  module
	searches for a solution within the  space of $N\times N$ images where 
	$N\times N$ is the target image size.    In the case of  scale-$(j)$ 
	reconstruction,   
	the optimization module searches for 
	the solution within   a restricted space of images denoted 
	${\cal S}(N,j)$.  Here ${\cal S}(N,j)$  denotes space of 
	$N\times N$  images obtained by interpolation from 
	$\frac{N}{2^j}\times \frac{N}{2^j}$  images via a $2^j$ interpolation. 
	Note that the images  in ${\cal S}(N,j)$  are $2^j$ times smoother.
	The reconstruction operations with this constrained optimization
	employing non-adaptive and adaptive regularizations are given in Figures
	\ref{fig:basicexp}, and \ref{fig:adreggexp} respectively.
	\begin{figure}[ht]
		\centering
		\includegraphics[width= 0.45\textwidth]{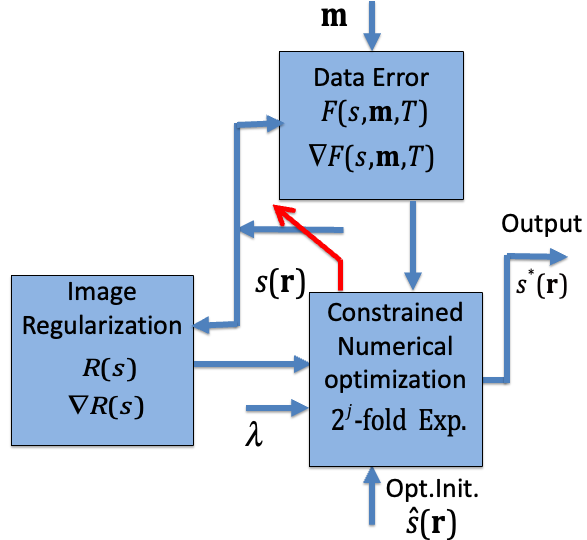}
		\caption{Schematic of  scale-$(j)$ non-adaptive regularized image 
		reconstruction	}
		\label{fig:basicexp}
	\end{figure}
	\begin{figure}[ht]
		\centering
		\includegraphics[width= 0.6\textwidth]{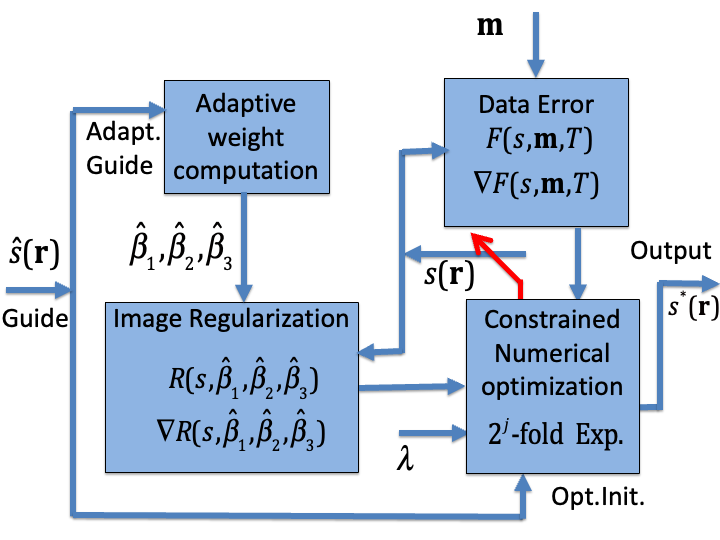}
		\caption{Schematic of scale-$(j)$  basic adaptive regularized 
		reconstruction}
		\label{fig:adreggexp}
	\end{figure}
	These reconstruction operations are symbolically represented as follows:
	\begin{align}
		\label{eq:recphsymbnaexp}
		s^{(j)} =   \mathlarger{\bar{\cal M}}^{(j)}\Big(
		\scalebox{.7}[1.0]{\bf Data}=\{T, {\bf m}\},  \;\;
		\lambda \Big)  
	\end{align}
	\begin{align}
		\label{eq:recphsymbaexp}
		s^{(j)} =   \mathlarger{\bar{\cal M}_{ad}}^{(j)}\Big(
		\scalebox{.7}[1.0]{\bf Init}
		=\hat{s},  
		\scalebox{.7}[1.0]{\bf AdapWt} = 
		\mathlarger{{\cal M}_w}\Big(
		\scalebox{.7}[1.0]{\bf AdapGd}=\hat{s},   \tau \Big), 
		\;\;  
		\scalebox{.7}[1.0]{\bf Data}=\{T, {\bf m}\},  \;\;
		\lambda \Big)  
	\end{align}
	With this,  our multiresolution scheme is constructed as given in
	Figure \ref{fig:schememr1}.  The first  block labelled as
	``scale J non-adaptive reconstruction" denotes a reconstruction method  
	where a 
	standard non-adaptive regularization is used (regularization defined in the 
	equation \eqref{eq:reghs}), which is
	represented by Figure \ref{fig:basicexp} with $j=J$.   This is symbolically 
	denoted as
	\begin{align}
		\label{eq:recphsymbaexpJ}
		s^{(J)} =   \mathlarger{\bar{\cal M}}^{(J)}\Big(
		\scalebox{.7}[1.0]{\bf Data}=\{T, {\bf m}\},  \;\;
		\lambda \Big).
	\end{align}
	The output $s^{(J)}$  becomes the adaptation guide to the next block, which 
	performs  a  scale-$(J-1)$ adaptive  
	regularized reconstruction  by searching  for the minimum  of corresponding
	adaptive cost within  ${\cal S}(N,J-1)$.   This block denotes the 
	reconstruction represented by the
	schematic of Figure \ref{fig:adreggexp} with $j=J-1$.   This operation is 
	represented by
	\begin{align}
		\label{eq:recphsymbaexpJ-1}
		s^{(J-1)} =   \mathlarger{\bar{\cal M}_{ad}}^{(J-1)}\Big(
		\scalebox{.7}[1.0]{\bf Init}
		=s^{(J)},  
		\scalebox{.7}[1.0]{\bf AdapWt} = 
		\mathlarger{{\cal M}_w}\Big(
		\scalebox{.7}[1.0]{\bf AdapGd}=s^{(J)},   \tau \Big), 
		\;\;  
		\scalebox{.7}[1.0]{\bf Data}=\{T, {\bf m}\},  \;\;
		\lambda \Big).
	\end{align}
	The output   $s^{(J-1)}$  becomes the  adaptation guide for the 
	scale-$(J-2)$ adaptive regularization
	reconstruction.   
	This process is continued until the scale-$(0)$ reconstruction  is 
	obtained,  
	which can be
	expressed  by the following equation for $j=J-2,\ldots 1, 0$:
	\begin{align}
		s^{(j)} =    {\mathlarger{\bar{\cal M}_{ad}}^{(j)}}\Big(
		\scalebox{.7}[1.0]{\bf Init}
		={s}^{(j+1)},   
		\scalebox{.7}[1.0]{\bf AdapWt} =
		\label{eq:recph12mrloop}
		\mathlarger{{\cal M}_w}\Big(
		\scalebox{.7}[1.0]{\bf AdapGd}= {s}^{(j+1)},   \tau \Big), 
		& \;\;\;\;\;\;\;\; \scalebox{.7}[1.0]{\bf Data}=\{T, {\bf m}\},  \;\;
		\lambda \Big)
	\end{align}
	These steps are represented by the remaining blocks in Figure 
	\ref{fig:schememr1}.  The last block,  performs scale-$(0)$  adaptive
	regularized reconstruction, where the optimization block search for 
	the solution within  ${\cal S}(N,0)$.  Note that ${\cal S}(N,0)$ is the set
	of $N\times N$ images without any interpolation constraint.  In other words,
	the reconstruction is  the same as the one given in Figure \ref{fig:adregg}.
	This means that the label `scale-$(0)$' is redundant and hence 
	we skipped it in the last block of Figure \ref{fig:schememr1}.
	
	Now, an analysis of  this multiresolution scheme is in order.
	In the series of reconstructions of 
	Figure \ref{fig:schememr1}, the first module
	uses non-adaptive regularization. 
	However, this is not a serious problem,
	because,  the target for this module is a reconstruction that is $2^{J}$ 
	smoother
	than the final reconstruction.  Other than this,  every reconstruction 
	module 
	receives
	an adaptation guide that was obtained by another adaptively regularized 
	reconstruction module.  Each reconstruction module except the first
	is able to discriminate high derivative values that are part of the image
	structure against high derivative values arising from noise and 
	oscillations;
	this is achieved by  making use of the information extracted from the 
	reconstruction  output of  previous module. 
	Hence, as we proceed through the modules, the resolution of the series of 
	reconstructions, 
	$\{s^{(j)}, j=J-1,J-2,J-3,\ldots, 2,1, 0\}$  improves significantly and 
	becomes 
	better than the resolution that
	can be obtained by any standard regularization method.  In particular, the 
	final reconstruction
	$s^{(0)}$ will have a significantly improved quality compared to the 
	reconstruction
	that can be obtained by any standard regularization method. 
	\begin{figure}[ht]
		\centering
		\includegraphics[width= 0.95\textwidth]{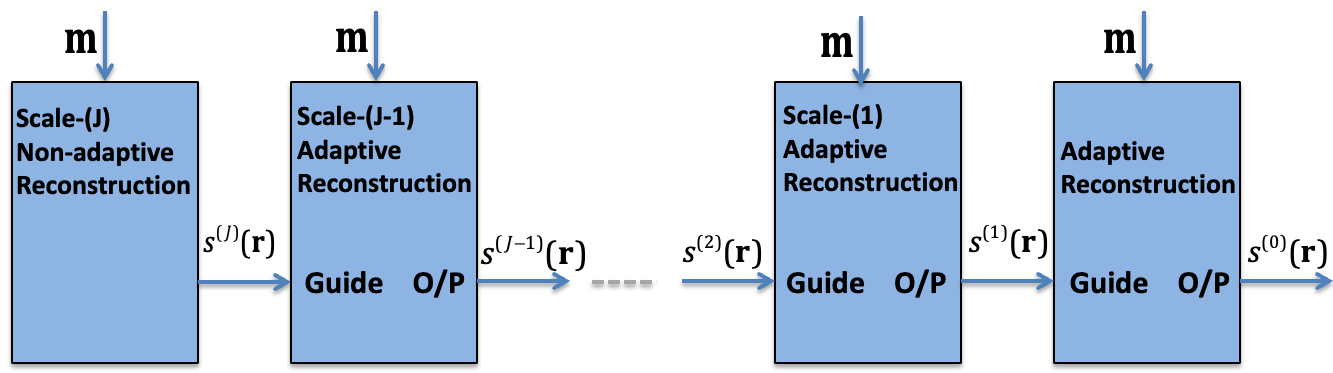}
		\caption{Schematic of multi-resolution based  adaptive regularized  
			reconstruction}
		\label{fig:schememr1}
	\end{figure}

	Now,  as the final reconstruction, $s^{(0)}(\r)$, will have an improved 
	reconstruction, 
	it can also be used to get improved adaptive weights, $\hat{\beta}_{1},  
	\hat{\beta}_{2}, 
	\hat{\beta}_{3}$, which in turn, can be used to get a further improved 
	reconstruction
	by performing an  adaptive  regularized reconstruction without scale 
	constraint (reconstruction
	scheme of Figure \ref{fig:adregg}).   This argument naturally
	leads  to another series of  adaptive regularized reconstructions,
	where the reconstructed 
	output
	of one module becomes the adaptation guide for the next module.
	This can be symbolically expressed, as given below:     
	\begin{align}
		s_{k+1} =   &\mathlarger{{\cal M}_{ad}}\Big(
		\scalebox{.7}[1.0]{\bf Init}
		={s}_k,  \;\ 
		\scalebox{.7}[1.0]{\bf AdapWt} = 
		\mathlarger{{\cal M}_w}\Big(
		\scalebox{.7}[1.0]{\bf AdapGd}={s}_k,   \tau \Big), 
		\label{eq:recph12iter}
		\;\;\;\;\;\;\;\;
		\scalebox{.7}[1.0]{\bf Data}=\{T, {\bf m}\},  \;\;
		\lambda \Big).
	\end{align}
	We set ${s}_0=s^{(0)}$.   Note that the main difference between
	the series given in the equation 	\eqref{eq:recph12mrloop}   and 
	\eqref{eq:recph12iter} is the following:  the reconstructions of the
	equation \eqref{eq:recph12iter} work on the final resolution without
	any interpolation constraint, whereas,  the reconstructions of the
	equation \eqref{eq:recph12mrloop} work on gradually increasing
	resolutions  with interpolation constraints.
	Under certain technical conditions, 
	the above iteration  leads to the solution of the following joint 
	minimization
	problem:
	\begin{equation}
		(s^*,\beta_1^*, \beta_2^*, \beta_3^*)   = 
		\underset{s,\beta_1, \beta_2, \beta_3}{\operatorname{argmin}} 
		F(s,T,{\bf m}) + \lambda R(s,\beta_1,\beta_2,\beta_3)  
		+
		\lambda\tau L(\beta_1,\beta_2, \beta_3)
	\end{equation}
	In other words, the sequence $\{s_{k}\}$  converges to $s^*$ of the above 
	minimization
	problem.  This completes the description of the proposed method.
	We call  our method  the Hessian-based combined order regularization
	with optimal spatial adaptation (H-COROSA).  The pre-cursor version,  
	COROSA 
	\cite{COROSA}, 
	is obtained by restricting the weights for  both canonical second 
	derivatives
	to be equal.

	\section{Results}
	\label{sec:exp}
	
	MRI  is non-invasive and  hence widely adopted in medical   diagnosis, and 
	biological investigations. However,  there is a serious need
	to reduce the scan time in MRI   to increase the frame rate and avoid 
	motion 
	artifacts. 
	In this viewpoint,  we note that an MRI system measures samples of the 
	Fourier 
	transform of a cross-sectional plane along a
	predefined path determined by the imaging system, and this path is called 
	the 
	sampling trajectory.  The fraction of samples measured 
	by the trajectory can be controlled  through some parameters,  and to 
	reduce 
	the scan time, this fraction   has to be kept low.
	Unfortunately, this leads to serious artifacts if traditional 
	interpolation-based
	reconstruction methods are used.  Recent  derivative-based regularization
	methods (also known as compressive sensing methods),  and deep learning 
	based 
	methods,  are believed to  reduce the artifacts,
	and the power of such modern methods is measured by their ability to  
	recover  
	sufficient resolution  from a low  fraction of samples. 
	Further, there is also demand on  modern reconstruction methods to recover 
	sufficient resolution from samples with a
	low signal-to-noise ratio (SNR). This   is because low-cost MRI systems 
	have low 
	magnetization strength   and hence the measured samples will have 
	low SNR;   also,  some structures in the imaging specimen may have low 
	magnetizability and hence the measured samples will
	have low SNR.   In this regard,  our goal in this section is to study the 
	ability of the proposed regularization
	to reconstruct images  accurately under the conditions of low sampling 
	fraction 
	and high noise in comparison with state-of-the-art regularization
	methods.     Further, we also intend to compare   the proposed    method 
	with  
	currently popular deep learning-based paradigms. Such deep learning based 
	approaches claim that they do not rely on any ad hoc prior and their 
	claimed 
	performance improvement over regularization methods have made them popular 
	in 
	several fields. At the same time, it is necessary to note that such 
	frameworks 
	require large training sets of appropriate images and they have been 
	observed 
	to be unstable under certain conditions \cite{Deep_Instability}.  
	Nevertheless,
	it is important to have a perspective on how such methods perform  compared
	to  the proposed method.
	
	With these stated objectives, we set up the experiments in two parts: the 
	first 
	part compares H-COROSA with existing derivative-based regularization 
	methods, 
	and the second part deals with the comparison against deep 
	learning based methods. With both sets of experiments,
	we compare the reconstruction performance using two quantitative scores: 
	Signal 
	to Noise Ratio (SNR) and Structurally Similarity Index 
	(SSIM) \cite{ssim}. These are widely used in image restoration literature 
	and 
	are useful in the objective evaluation of reconstruction quality. 
	The experimental setup for each set is described in the corresponding 
	subsections below.  Before we proceed, we note that sampling
	fraction is called sampling density in the image reconstruction literature 
	and 
	we will use the same term.  Further,  in this section, 
	a sampling trajectory with parameters chosen for a specific sampling 
	density is called a sampling scheme.

	\subsection{Comparison with derivative-based regularization methods}
	

	\begin{figure}[htbp]
		\centering
		\includegraphics[width= 0.98\textwidth]{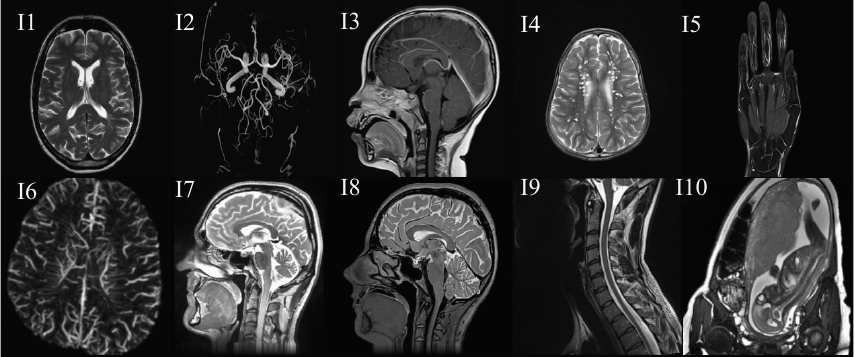}
		\caption{MRI Reference Images I1-I10.}
		\label{fig:mri_data}
	\end{figure}
	
	\begin{figure}[htbp]
		\centering
		\includegraphics[width= 0.98\textwidth]{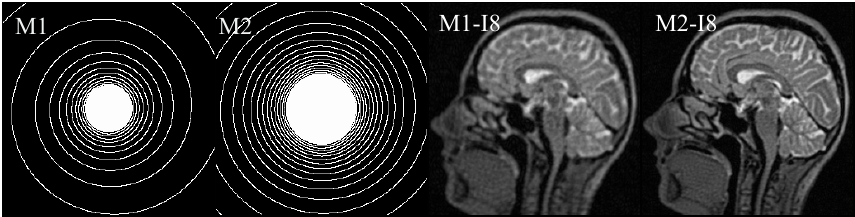}
		\caption{Spiral sampling trajectories 
			$M1$(10\%) and $M2$(20\%), Zero filled Fourier inversions of 
			undersampled 
			$I8$ with samples from $M1$(M1-I8) and $M2$(M2-I8).}
		\label{fig:mri_traj}
	\end{figure}

	We consider the following   regularization methods:
	\begin{enumerate}
		\item
		Second-order Total Variation (TV2) \cite{TV1_Rudin,Scherzer_tv2_98}: 
		this belongs to the class of non-adaptive regularization method 
		depicted in
		Figure 	\ref{fig:basic}  with $R(g)$ replaced by the form given 
		in Equation  \eqref{eq:regtv2}.
		This regularization
		still remains one of the most popular and robust  regularization 
		functionals in image restoration. Lustig et. al. 
		\cite{MRI_Sparse_Lustig2007sparse} used TV in their seminal work on 
		undersampled MRI reconstruction. 
		\item
		Hessian-Schatten (HS) norm regularization \cite{HS_13}: 
		This again belongs to the class of non-adaptive regularization depicted 
		in
		Figure 	\ref{fig:basic}  with $R(g)$ replaced by the form given 
		in Equation  \eqref{eq:reghs}.  
		\item
		Total Generalized Variation (TGV) \cite{TGV2}: Total Generalized 
		Variation 
		(TGV) has been proposed 
		as a generalization of Total Variation (TV).  Only a specific form of  
		TGV, known 
		as   second-order TGV is practicable for our problem and hence we choose
		this form.
		\item
		Combined Order Regularization with Optimal Spatial Adaptation (COROSA) 
		\cite{COROSA}: COROSA is a 
		spatially weighted multi-order derivative-based regularization method 
		for image 
		restoration. This method 
		has been applied to Total Internal Reflection Microscopy (TIRF) 
		restoration and 
		MRI reconstruction
		with a performance that  is better than the  other derivative-based 
		regularization 
		methods. This is 
		the precursor of the method developed in this paper, H-COROSA.
	\end{enumerate}
	We used a set of ten MRI images (I1-I10) shown in Figure \ref{fig:mri_data} 
	as 
	the reference models
	for comparing the proposed method with the methods listed above. 
	We intend to study the effect of both sampling density and sampling 
	trajectory 
	on the performance
	of various 
	reconstruction methods.  For evaluating performance under variation in 
	sampling 
	density, 
	we first consider a spiral sampling scheme;  we used two forms of spiral 
	with 
	their parameters
	adjusted such that the  trajectories cover 10\%  and 20\% of total Fourier 
	samples.
	These trajectories are
	given in Figure \ref{fig:mri_traj} with labels $M_{1}$  and $M_{2}$.   The 
	required 
	measurements were 
	simulated  by Fourier
	transforming  the model images as a first step, and then picking the 
	Fourier samples that fall 
	in the trajectory. 
	The spiral trajectory is an important case  because it is considered to be 
	the most 
	efficient among other trajectories and its theoretical background is well 
	established.  To make the measurements 
	realistic,  we also simulated the effect of thermal noise in measurements 
	by 
	adding Additive White Gaussian Noise (AWGN) to the simulated Fourier 
	samples 
	\cite{MRI_Haacke1999magnetic}. The noise variance was set using the 
	strategy of 
	Ravishankar et. al. \cite{MRI_Bresler, COROSA}, such that the noisy 
	reconstruction from simple Fourier inversion of the full set of samples 
	with 
	added noise, gives a chosen Peak Signal to Noise Ratio (PSNR) value in 
	comparison with the clean original image. 
	We chose the PSNR to be  20dB for the measurements coming from  both
	spiral 
	trajectories. With two trajectories and ten model  images, twenty noisy 
	measurement sets were simulated  and given as input to all reconstruction 
	methods to be evaluated.
	Figure \ref{fig:mri_traj} also shows   the inverse transform of the two 
	sets of 
	samples
	for  a specific model image, where we filled in missing sample locations 
	with 
	zeros.  
	Table \ref{table:ssim_snr_spiral} compares the reconstruction scores of the 
	proposed  methods
	with other regularization-based methods.
	\begin{table}[htbp]
		\centering
		\caption{MRI Reconstruction Scores with Spiral Trajectory: 10\% and 
		20\% 
			Sampling Densities}
		\renewcommand{\arraystretch}{0.95}
		\begin{tabular}[t]{ccccc}
			\toprule
			\multirow{2}{*}{I1} &\multicolumn{2}{c}{SSIM} 
			&\multicolumn{2}{c}{SNR} 
			\\  &10\% &20\% &10\% &20\%\\
			\midrule
			H-COR. &\textbf{.967}&\textbf{.983}&\textbf{24.57}&27.94\\
			COR. &.966&\textbf{.983}&24.09&\textbf{28.12}\\
			TGV2 &.959&.979&22.86&26.68\\
			TV2  &.955&.972&22.07&24.71\\
			HS   &.955&.976&22.14&25.78\\
			\toprule
			\multirow{2}{*}{I2} &\multicolumn{2}{c}{SSIM} 
			&\multicolumn{2}{c}{SNR} 
			\\  &10\% &20\% &10\% &20\%\\
			\midrule
			H-COR. &\textbf{.994}&\textbf{.998}&\textbf{25.02}&\textbf{30.74}\\
			COR. &.992&\textbf{.998}&23.58&30.45\\
			TGV2 &.982&.995&20.20&26.06\\
			TV2  &.978&.990&18.95&22.68\\
			HS   &.978&.993&19.09&24.54\\
			\toprule
			\multirow{2}{*}{I3} &\multicolumn{2}{c}{SSIM} 
			&\multicolumn{2}{c}{SNR} 
			\\  &10\% &20\% &10\% &20\%\\
			\midrule
			H-COR. &\textbf{.921}&\textbf{.968}&\textbf{20.98}&\textbf{25.17}\\
			COR. &.898&.962&19.44&24.58\\
			TGV2 &.879&.952&18.33&23.12\\
			TV2  &.879&.940&18.00&21.50\\
			HS   &.882&.948&18.18&22.62\\
			\toprule
			\multirow{2}{*}{I4} &\multicolumn{2}{c}{SSIM} 
			&\multicolumn{2}{c}{SNR} 
			\\ &10\% &20\% &10\% &20\%\\
			\midrule
			H-COR. &\textbf{.945}&.974&\textbf{20.14}&24.12\\
			COROSA &.938&\textbf{.975}&19.48&\textbf{24.38}\\
			TGV2 &.888 &.926&18.50&22.82\\
			TV2  &.922&.958&18.16&21.00\\
			HS   &.923&.964&18.19&21.98\\
			\toprule
			\multirow{2}{*}{I5} &\multicolumn{2}{c}{SSIM} 
			&\multicolumn{2}{c}{SNR} 
			\\ &10\% &20\% &10\% &20\%\\
			\midrule
			H-COR. &\textbf{.987}&.995&\textbf{21.36}&26.99\\
			COR. &.986&\textbf{.996}&21.21&\textbf{27.43}\\
			TGV2 &.977&.984&18.62&19.96\\
			TV2  &.972&.984&17.15&19.36\\
			HS   &.972&.987&17.18&20.98\\
			\bottomrule
		\end{tabular}
		\renewcommand{\arraystretch}{0.95}
		\begin{tabular}[t]{ccccc}
			\toprule
			\multirow{2}{*}{I6} &\multicolumn{2}{c}{SSIM} 
			&\multicolumn{2}{c}{SNR} 
			\\  &10\% &20\% &10\% &20\%\\
			\midrule
			H-COR. &.937&.983&19.39&25.51\\
			COR. &\textbf{.940}&\textbf{.985}&\textbf{19.71}&\textbf{26.12}\\
			TGV2 &.926&.984&18.33&25.80\\
			TV2  &.921&.982&18.22&25.52\\
			HS   &.922&.982&18.32&25.57\\
			\toprule
			\multirow{2}{*}{I7} &\multicolumn{2}{c}{SSIM} 
			&\multicolumn{2}{c}{SNR} 
			\\  &10\% &20\% &10\% &20\%\\
			\midrule
			H-COR. &\textbf{.976}&\textbf{.994}&\textbf{31.14}&37.83\\
			COR. &.968&.992&29.44&36.46\\
			TGV2 &.968&\textbf{.994}&29.08&\textbf{38.20}\\
			TV2  &.965&\textbf{.994}&28.98&37.99\\
			HS   &.966&\textbf{.994}&29.25&38.09\\
			\toprule
			\multirow{2}{*}{I8} &\multicolumn{2}{c}{SSIM} 
			&\multicolumn{2}{c}{SNR} 
			\\  &10\% &20\% &10\% &20\%\\
			\midrule
			H-COR. &\textbf{.878}&\textbf{.942}&\textbf{17.42}&20.91\\
			COR. &.851&.937&16.73&\textbf{20.98}\\
			TGV2 &.853&.939&16.52&20.56\\
			TV2  &.843&.932&16.46&20.32\\
			HS   &.847&.934&16.60&20.48\\
			
			\toprule
			\multirow{2}{*}{I9} &\multicolumn{2}{c}{SSIM} 
			&\multicolumn{2}{c}{SNR} 
			\\  &10\% &20\% &10\% &20\%\\
			\midrule
			H-COR. &\textbf{.927}&.964&\textbf{18.34}&22.24\\
			COR. &.921&\textbf{.965}&17.95&\textbf{22.40}\\
			TGV2 &.900&.957&16.73&21.19\\
			TV2  &.895&.952&16.56&20.56\\
			HS   &.896&.952&16.65&20.69\\
			\toprule
			\multirow{2}{*}{I10} &\multicolumn{2}{c}{SSIM} 
			&\multicolumn{2}{c}{SNR} 
			\\ &10\% &20\% &10\% &20\%\\
			\midrule
			H-COR. &\textbf{.902}&\textbf{.953}&\textbf{21.20}&24.95\\
			COR. &.892&.952&20.68&\textbf{25.01}\\
			TGV2 &.886&.954&20.10&24.91\\
			TV2  &.880&.949&19.83&24.50\\
			HS   &.882&.949&19.92&24.59\\
			\bottomrule
		\end{tabular}
		\label{table:ssim_snr_spiral}
	\end{table}

	
	
	From Table \ref{table:ssim_snr_spiral},  it is
	evident that 
	H-COROSA  restores 
	resolution  better than other methods, giving higher reconstruction 
	scores.  
	While the 
	proposed method
	outperforms the compared methods in both types of sample sets,  the 
	superiority 
	of the proposed
	more evident if we compare reconstructions obtained from 10\% sample sets.  
	It shows the
	advantage of adaptively combining the canonical second derivatives and 
	gradients.  The closest
	competitor to the proposed methods is COROSA,  which is the pre-cursor of 
	the 
	proposed method.
	Figure \ref{fig:mri_reconstruction_mri2_spiral_sr_1_region1} shows a 
	portion of 
	the reconstruction result of $I2$ from 10\% spiral samples, where H-COROSA 
	is 
	clearly able to restore a fine image structure that is missed by all other 
	methods. This is evident in the corresponding scanline where a clear rise 
	in 
	intensity is visible in H-COROSA scanline, in alignment with the original 
	intensity profile. Figure 
	\ref{fig:mri_reconstruction_mri11_spiral_sr_2_region1} shows a 
	reconstructed 
	region from $I10$ with 20\% spiral samples. Again, H-COROSA reconstruction 
	is 
	more clear and smooth, while distortions  
	are visible in the closest competitor, 
	the  COROSA construction.   This 
	clearly 
	confirms the advantage
	of adaptively combining three terms, i.e.,  the advantage of combining two 
	canonical  
	second
	derivatives and the gradient.  Recall that COROSA only uses two terms,  
	namely 
	the gradient and
	the  norm of the Hessian.   The corresponding scanline shows that 
	H-COROSA aligns
	with the original intensity profile whereas the closest competitor COROSA 
	suffers from oscillations. Figure 
	\ref{fig:mri_reconstruction_mri3_spiral_sr_2_region1} shows a reconstructed 
	portion of $I3$ where H-COROSA not only restores fine details  but also has 
	significantly reduced artifacts.
	The intensity profile along the selected scan line also shows that H-COROSA 
	not 
	only tracks the clean original intensity  but also avoids oscillations. In 
	effect, we can summarize from the first experiment that H-COROSA  provides 
	significant advantages over other state-of-the-art regularization methods 
	in 
	MRI reconstruction, when the sampling density is low.  
	
	\begin{figure}[htbp]
		\centering
		\includegraphics[width= 
		0.8\textwidth]{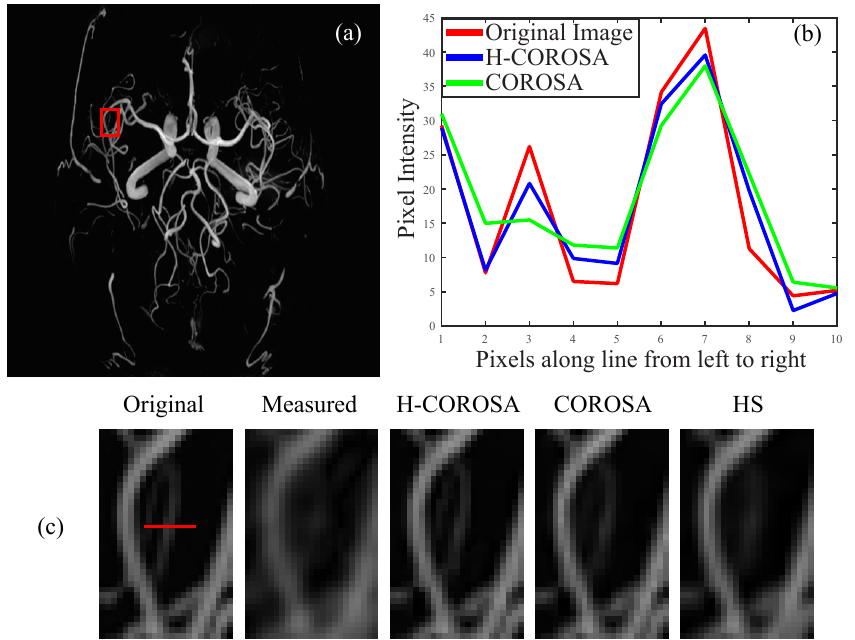}
		\caption{Reconstruction of $I2$ from 10\% spiral samples: (a) Original 
		$I2$ 
			with selected region (b) Intensity profile along a scanline in the 
			selected 
			region (c) Comparison of selected region from reconstructions along 
			with 
			the scanline location.  Here `measured'  denotes  zero-filled 
			inverse DFT of
			measured samples.}
		\label{fig:mri_reconstruction_mri2_spiral_sr_1_region1}
	\end{figure}
	
	\begin{figure}[htbp]
		\centering
		\includegraphics[width= 
		0.8\textwidth]{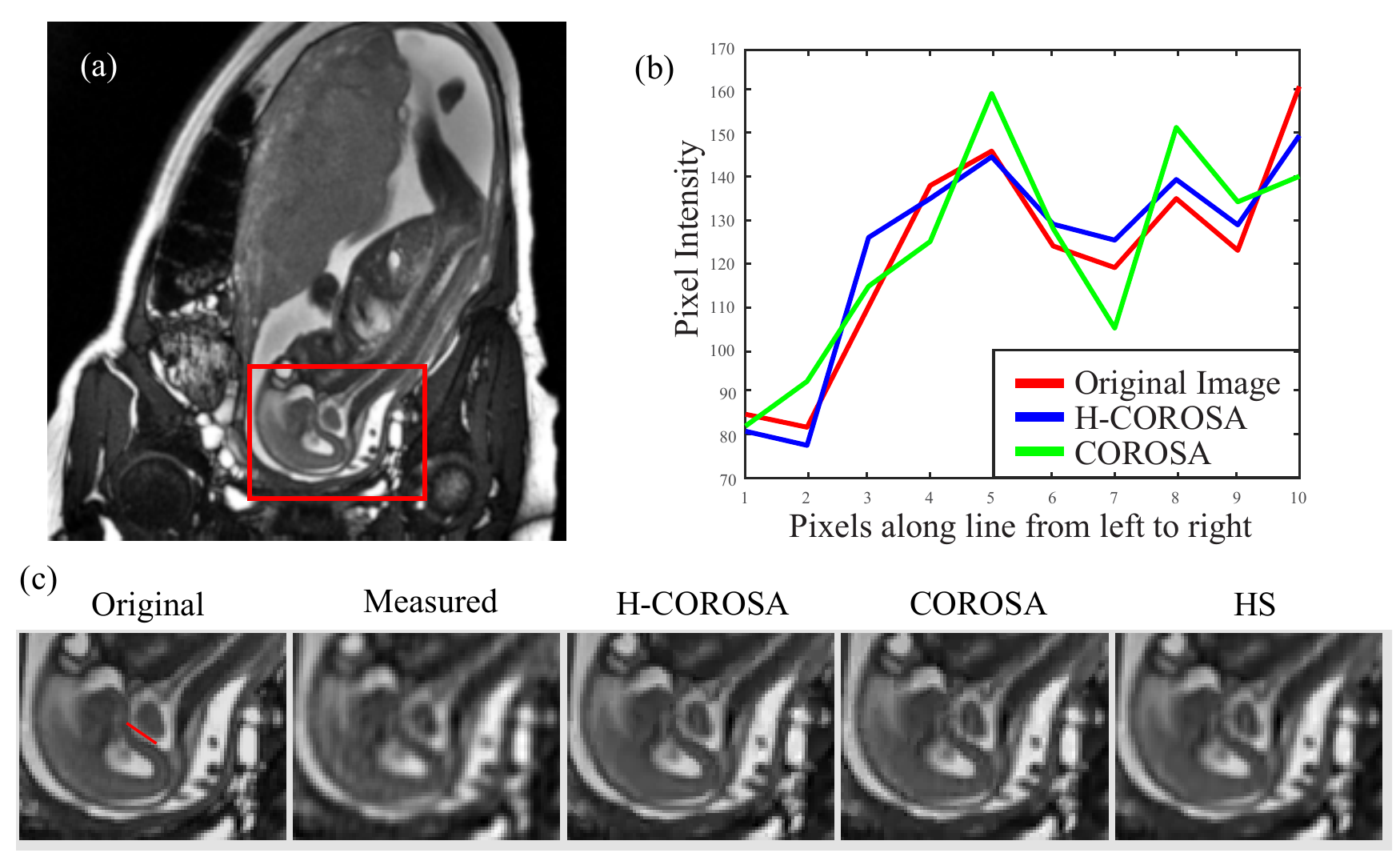}
		\caption{Reconstruction of $I10$ from 20\% spiral samples: (a) Original 
			$I10$ with selected region (b) Intensity profile along a scanline 
			in the 
			selected region (c) Comparison of selected region from 
			reconstructions 
			along with the scanline location. Here `measured'  denotes  
			zero-filled inverse DFT of
			measured samples.}
		\label{fig:mri_reconstruction_mri11_spiral_sr_2_region1}
	\end{figure}
	
	\begin{figure}[htbp]
		\centering
		\includegraphics[width= 
		0.8\textwidth]{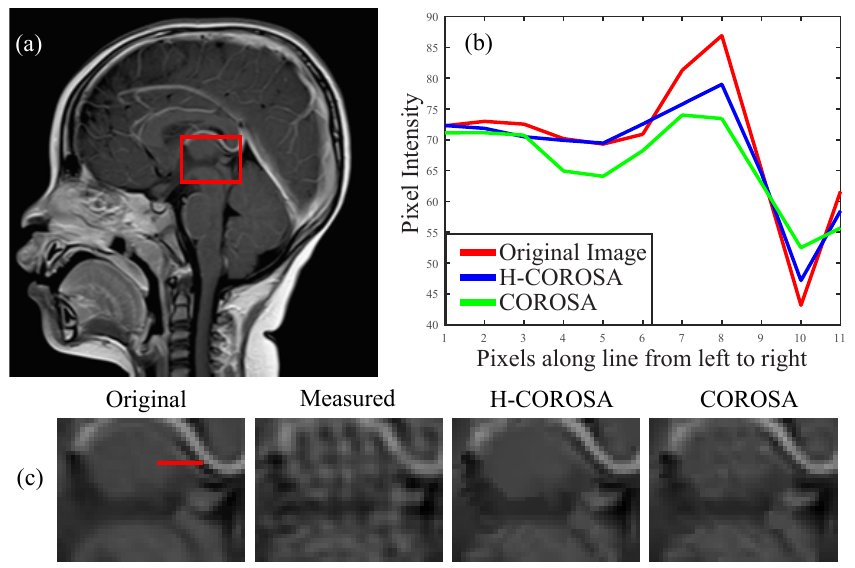}
		\caption{Reconstruction of $I3$ from 20\% spiral samples: (a) Original 
		$I2$ 
			with selected region (b) Intensity profile along a scanline in the 
			selected 
			region (c) Comparison of selected region from reconstructions along 
			with 
			the scanline location.  Here `measured'  denotes  zero-filled 
			inverse DFT of
			measured samples.}
		\label{fig:mri_reconstruction_mri3_spiral_sr_2_region1}
	\end{figure}

	Naturally, the next question is whether this advantage of H-COROSA is 
	trajectory specific 
	and so we consider two forms of sampling trajectories:  radial trajectory  
	given 
	the figure
	\ref{fig:mri_reconstruction_mri9_radial_tsp_sr_1_overview}.R1,  and a form 
	of 
	quasi-random
	trajectory known as Traveling Salesman Problem (TSP) trajectory given in 
	the 
	Figure 
	\ref{fig:mri_reconstruction_mri9_radial_tsp_sr_1_overview}.T1.  The TSP  
	trajectory proposed
	by Chauffert et al \cite{MRI_Chauffert},   was
	inspired by the recent observation that random sampling leads to improved 
	reconstruction \cite{CS_candes2006robust, MRI_Sparse_Lustig2007sparse, 
		CS_MRI_knoll2011adapted, CS_MRI_krahmer2012beyond, 
		MRI_TSP_chauffert2013travelling}.
	As a truly  random trajectory is not physically realizable,  the authors 
	constructed the trajectory
	given in the Figure by solving the well-known traveling salesman problem. 
	We 
	set the sampling density of $R1$ to 10\% and $T1$ to 12\%, since we are 
	interested in performance under low sampling density. The noise variance 
	was 
	set such that the  PSNR is 20dB, as done in the previous experiment. Figure 
	\ref{fig:mri_reconstruction_mri9_radial_tsp_sr_1_overview} also shows the 
	loss 
	of details using  
	zero-filled Fourier inversions. The evaluation scores of the 
	reconstructions 
	obtained by various methods
	from the sample set corresponding to these trajectories are given in Table  
	\ref{table:ssim_snr_tsp}.
	\begin{figure}[htbp]
		\centering
		\includegraphics[width= 
		0.8\textwidth]{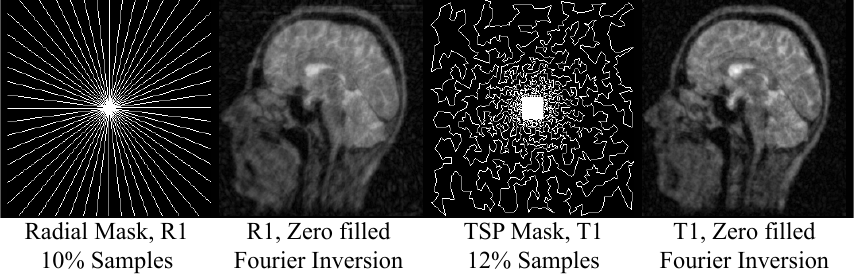}
		\caption{Radial sampling trajectory  $R1$ and TSP trajectory $T1$ 
		along 
				with zero filled Fourier inversions of undersampled $I8$ image 
				with samples 
				from corresponding trajectories.}
		\label{fig:mri_reconstruction_mri9_radial_tsp_sr_1_overview}
	\end{figure}
	
	\begin{table}[htbp]
		\centering
		\caption{MRI Reconstruction: Scores with TSP and  Radial Sampling}
		\renewcommand{\arraystretch}{0.95}
		\begin{tabular}[t]{ccccc}
			\toprule
			\multirow{2}{*}{I1} &\multicolumn{2}{c}{SSIM} 
			&\multicolumn{2}{c}{SNR} 
			\\  &TSP &Rad. &TSP &Rad.\\
			\midrule
			H-COR. &\textbf{.970}&\textbf{.963}&\textbf{25.66}&\textbf{24.03}\\
			COR. &.968&.961&25.24&23.50\\
			TGV2 &.956&.946&23.33&21.42\\
			TV2  &.954&.943&22.68&21.29\\
			HS   &.955&.944&22.84&21.43\\
			\toprule
			\multirow{2}{*}{I2} &\multicolumn{2}{c}{SSIM} 
			&\multicolumn{2}{c}{SNR} 
			\\  &TSP &Rad. &TSP &Rad.\\
			\midrule
			H-COR. &\textbf{.997}&\textbf{.994}&\textbf{28.72}&\textbf{25.04}\\
			COR. &.996&.992&28.05&23.77\\
			TGV2 &.986&.951&22.11&17.75\\
			TV2  &.983&.973&20.71&18.50\\
			HS   &.983&.974&20.98&18.73\\
			\toprule
			\multirow{2}{*}{I3} &\multicolumn{2}{c}{SSIM} 
			&\multicolumn{2}{c}{SNR} 
			\\  &TSP &Rad. &TSP &Rad.\\
			\midrule
			H-COR. &\textbf{.916}&\textbf{.894}&\textbf{20.93}&\textbf{18.77}\\
			COR. &.886&.857&19.44&17.20\\
			TGV2 &.835&.747&17.78&15.04\\
			TV2  &.843&.789&17.59&15.49\\
			HS   &.851&.798&17.97&15.76\\
			\toprule
			\multirow{2}{*}{I4} &\multicolumn{2}{c}{SSIM} 
			&\multicolumn{2}{c}{SNR} 
			\\ &TSP &Rad. &TSP &Rad.\\
			\midrule
			H-COR. &\textbf{.958}&\textbf{.944}&\textbf{22.07}&\textbf{20.15}\\
			COROSA &.953&.940&21.76&19.94\\
			TGV2 &.930&.866&19.76&17.75\\
			TV2  &.929&.910&19.37&18.22\\
			HS   &.930&.911&19.46&18.31\\
			\toprule
			\multirow{2}{*}{I5} &\multicolumn{2}{c}{SSIM} 
			&\multicolumn{2}{c}{SNR} 
			\\  &TSP &Rad. &TSP &Rad.\\
			\midrule
			H-COR. &\textbf{.991}&\textbf{.980}&23.88&19.16\\
			COR. &.990&.979&\textbf{23.89}&\textbf{19.18}\\
			TGV2 &.982&.959&20.81&16.55\\
			TV2  &.979&.964&19.46&16.42\\
			HS   &.979&.965&19.54&16.50\\
			\bottomrule
		\end{tabular}
		\renewcommand{\arraystretch}{0.95}
		\begin{tabular}[t]{ccccc}
			\toprule
			\multirow{2}{*}{I6} &\multicolumn{2}{c}{SSIM} 
			&\multicolumn{2}{c}{SNR} 
			\\  &TSP &Rad. &TSP &Rad.\\
			\midrule
			H-COR. &\textbf{.911}&\textbf{.878}&\textbf{17.93}&\textbf{15.86}\\
			COR. &.859&.859&17.28&15.55\\
			TGV2 &.838&.806&15.25&13.70\\
			TV2  &.844&.806&15.20&13.65\\
			HS   &.846&.808&15.31&13.74\\
			\toprule
			\multirow{2}{*}{I7} &\multicolumn{2}{c}{SSIM} 
			&\multicolumn{2}{c}{SNR} 
			\\  &TSP &Rad. &TSP &Rad.\\
			\midrule
			H-COR. &\textbf{.965}&\textbf{.956}&\textbf{29.69}&\textbf{27.84}\\
			COR. &.949&.937&27.50&25.90\\
			TGV2 &.930&.914&25.70&24.04\\
			TV2  &.928&.910&25.58&24.02\\
			HS   &.933&.916&26.05&24.43\\
			\toprule
			\multirow{2}{*}{I8} &\multicolumn{2}{c}{SSIM} 
			&\multicolumn{2}{c}{SNR} 
			\\  &TSP &Rad. &TSP &Rad.\\
			\midrule
			H-COR. &\textbf{.880}&\textbf{.840}&\textbf{17.82}&\textbf{16.37}\\
			COR. &.847&.808&16.97&15.64\\
			TGV2 &.823&.797&16.28&15.19\\
			TV2  &.822&.787&16.17&15.13\\
			HS   &.827&.792&16.33&15.27\\
			\toprule
			\multirow{2}{*}{I9} &\multicolumn{2}{c}{SSIM} 
			&\multicolumn{2}{c}{SNR} 
			\\  &TSP &Rad. &TSP &Rad.\\
			\midrule
			H-COR. &\textbf{.932}&\textbf{.900}&\textbf{19.15}&\textbf{16.35}\\
			COR. &.917&.885&18.28&15.54\\
			TGV2 &.877&.847&16.07&13.98\\
			TV2  &.871&.837&15.81&13.79\\
			HS   &.875&.841&15.99&13.98\\
			\toprule
			\multirow{2}{*}{I10} &\multicolumn{2}{c}{SSIM} 
			&\multicolumn{2}{c}{SNR} 
			\\ &TSP &Rad. &TSP &Rad.\\
			\midrule
			H-COR. &\textbf{.882}&\textbf{.841}&\textbf{20.83}&\textbf{18.55}\\
			COR. &.865&.818&20.25&17.92\\
			TGV2 &.827&.772&18.84&16.56\\
			TV2  &.831&.770&18.60&16.35\\
			HS   &.835&.775&18.77&16.49\\
			\bottomrule
		\end{tabular}
		\label{table:ssim_snr_tsp}
	\end{table}
	
	Table \ref{table:ssim_snr_tsp}  shows that H-COROSA performs better than 
	all methods with both trajectories, with a higher improvement in score over 
	other 
	methods in the case of TSP sampling. This table clearly confirms that
	H-COROSA gives better reconstruction over other regularization  methods 
	with 
	low sampling density, irrespective of the utilized trajectory. Figure 
	\ref{fig:mri_reconstruction_mri9_tsp_sr_12_region1} shows a reconstructed 
	region from $I8$ with 12\% TSP sampling, where the restoration by H-COROSA 
	is 
	noticeable sharper than COROSA. The corresponding scanline shows that both 
	methods smooth out some image features on account of low sampling density 
	and 
	noise.  However, H-COROSA is able to restore sharper details aligning 
	better 
	with  the original intensity profile, compared to COROSA. In summary, the 
	first 
	set of experiments show that the proposed H-COROSA approach is able to 
	restore 
	image details better than all well-demonstrated regularization-based 
	methods in 
	the presence of 
	under-sampling and noise. This improvement  is reflected both in terms of 
	quantitative scores (SNR and SSIM) as well as visual  inspection of 
		reconstructions, where 
	H-COROSA is able to restore fine structures while eliminating 
	distortions 
	caused by  
	undersampling and noise. 
	
	\begin{figure}[htbp]
		\centering
		\includegraphics[width= 
		0.8\textwidth]{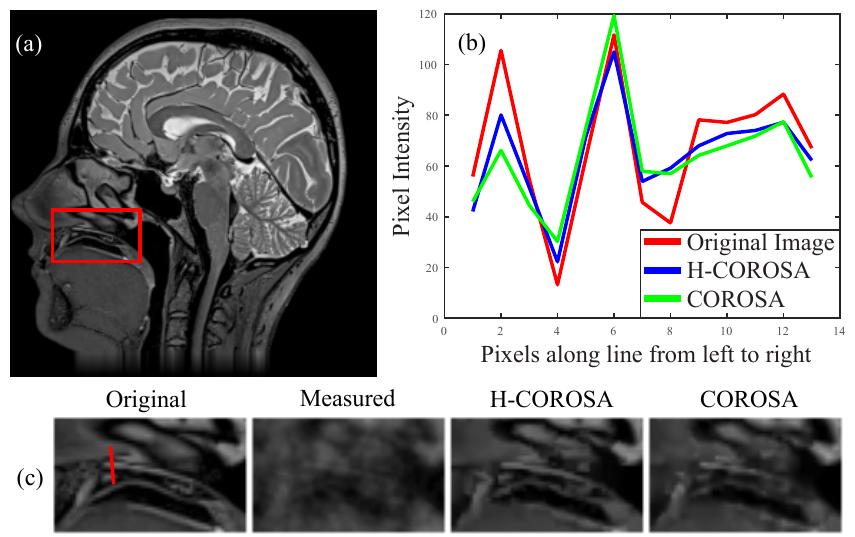}
		\caption{Reconstruction of $I8$ from 12\% tsp samples: (a) Original 
		$I8$ 
			with selected region (b) Intensity profile along a scanline in the 
			selected 
			region (c) Comparison of selected region from reconstructions along 
			with 
			the scanline location.  Here `measured'  denotes  zero-filled 
			inverse DFT of
			measured samples.}
		\label{fig:mri_reconstruction_mri9_tsp_sr_12_region1}
	\end{figure}
	
	\subsection{Comparison with deep learning methods}
	
	The goal of this section is to compare the proposed method with deep 
	learning 
	methods.  Deep learning methods
	for image recovery use layers of the so-called convolutional neural 
	networks 
	(CNN), where the term ``deep''
	signifies the stacking of several layers.  Many  ways of using  deep CNN 
	have
	been reported for image recovery.
	Deep CNNs have several thousands of parameters whose values need to be 
	determined by the process known as training.  There are two types of 
	training paradigm reported in the 
	literature:
	
	\noindent
	\underline{Sampling scheme dependent training:}
	In this type,  training requires data containing thousands of pairs,  
	where, 
	each pair contains a
	representative measured image  and the ground truth image that has 
	generated 
	the measurement set.
	From the training data, network parameters are determined  through 
	computational optimization such that
	the network approximately  produces the ground truth image as its output 
	from 
	the corresponding measurement
	as its input.   Clearly,   a trained network   becomes specific to the 
	sampling scheme.  If the sampling scheme
	is changed, the network has to be retrained even if the same class of 
	ground truth 
	images is involved.
	There is a subclass in this category of networks called   
	Generative-Adversarial Network 
	(GAN) architecture, where two 
	CNNs are trained simultaneously: a generator CNN  to generate the 
	reconstruction from noisy 
	under-sampled measurement and a discriminator CNN  to determine the quality 
	of 
	the reconstruction. 
	This results in a combined optimization between the two networks, towards 
	improving the quality of reconstruction. 
	Yang et al applied this scheme for MRI reconstruction \cite{DAGAN},  and 
	they 
	called the resulting network the 
	Deep De-Aliasing Generative Adversarial Networks (DAGAN).  
	
	\noindent
	\underline{Plug and Play  deep learning approach:}  This approach mimics 
	the 
	regularization  method,  where a deep CNN
	plays the role of regularization.  However,  the network does not work 
	directly 
	as   the regularization, but  indirectly
	as a  denoising network. To be more specific,  we first note that  image 
	reconstruction using derivative-based regularization
	is carried out by  splitting-based  numerical optimization methods.  A 
	splitting-based method
	obtains reconstruction by using computational steps involving data fitting 
	and 
	regularization parts independently
	and performs these steps cyclically until some convergence condition is 
	attained.  A crucial  computational 
	step in such schemes is denoising certain intermediate images using  
	regularization. Plug and Play  (PnP)
	deep learning methods   are constructed by replacing this denoising step 
	with
	deep CNN.  Depending on the type
	of splitting-based method used as the base method,  various types of PnP 
	methods 
	are formed. 
	A class of PnP methods  known  as  the ADMM-PnP method was constructed by 
	modifying  a 
	widely used splitting 
	method known as the Alternating Direction Method of Multipliers.   In 
	this 
	class,  
	the most successful one 
	was developed by Ryu et al \cite{PNP}.  We select this method for our 
	comparison.  For brevity,  we simply refer to  this
	method by PnP,  although there are various types of PnP methods.  Note that 
	the 
	training here is not specific to the sampling scheme.
	However, the training is specific to the least-squares difference between 
	zero-filled-in Fourier inversion of
	the measurement and the underlying ground truth image.

	For both types of methods listed above,   we make use of training data 
	used by the authors of DAGAN \cite{MICCAI13_dataset}. We also use the 
	recommended 
	settings/parameters for training DAGAN network, with the code  
	given by the 
	authors \cite{DAGAN_Code}.   In addition,  we also
	include the trained networks  posted by the authors of the PnP 
	method for comparison \cite{PNP_Code}. Only the best score from these 
	networks is updated as the
	reconstruction score for the PnP method. To evaluate  all methods   
	considered 
	in 
	this  
	experiment,  we select 
	six images from the testing set used by the authors of DAGAN.
	These images are shown in Figure 
	\ref{fig:mri_deep_learning_dataset} (T1-T6).  The figure also shows the two 
	sampling schemes that we
	used to simulate the measurements;  the first scheme with 10\% sampling 
	density 
	is shown in 
	the sub-figure R1 and the second with 20\% sampling density is shown in the 
	sub-figure 
	R2.   The sample locations for these schemes were selected randomly.  
	Unlike 
	the TSP sampling
	scheme used in the previous experiment which also has some randomness,   
	the 
	sampling schemes
	used here are not physically realizable.  Nevertheless,  this type of 
	scheme has 
	been  used 
	in compressive sensing literature \cite{MRI_Sparse_Lustig2007sparse}, and 
	hence 
	we use these 
	in order to get a perspective.  The sampling schemes R1 and R2, and the 
	test data 
	images,  T1-T6,
	together give 12 sample sets.  We evaluate the methods under consideration 
	using these sets
	without adding any noise, as well as, with noise added such that 
		the SNR is -0.5dB for 10\% sampling scheme, and  -1.5dB for 20\% 
		sampling 
		scheme with respect to zero filled Fourier inversions. 
	This makes a total of 24 sample sets.  Evaluation scores for 
	reconstructions from all methods 
	obtained using these
	sample sets without noise are given in  Table \ref{table:deep_noiseless}, 
	and the scores obtained using  noisy sample sets are shown in Table 
	\ref{table:deep_noisy}. 
	
	\begin{figure}[ht]
		\centering
		\includegraphics[width= 
		0.9\textwidth]{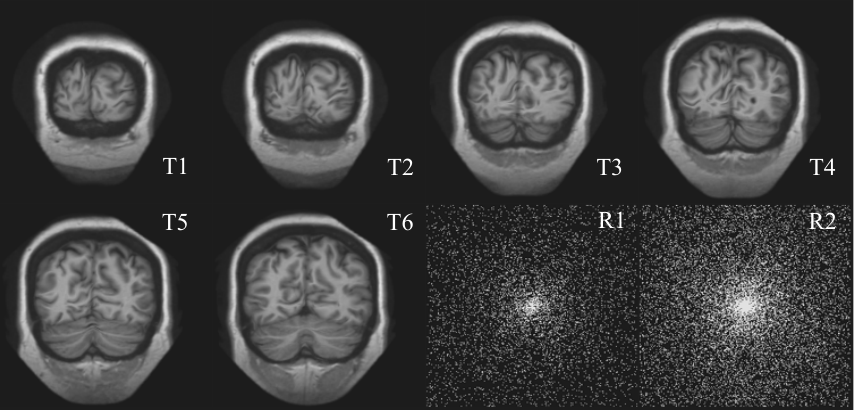}
		\caption{MRI Deep Learning Comparison Dataset: $T1$-$T4$, Random masks 
		$R1$ 
			(10\%) and $R2$ (20\%)}
		\label{fig:mri_deep_learning_dataset}
	\end{figure}
	

	
	Table \ref{table:deep_noiseless} confirms  that, in general,  H-COROSA 
	outperforms all methods 
	in terms of SSIM scores, with the precursor method, COROSA being the 
	closest 
	competitor.
	In terms of SNR scores,   H-COROSA again outperforms other methods in most 
	cases 
	with few
	exceptions where   DAGAN and PNP  have higher scores.  It should be however 
	emphasized
	that SNR measure is insensitive to artifacts containing artificial 
	sharpening,   
	while SSIM measure is  more sensitive to such artifacts and hence more  
	reliable. This means that
	in terms of preserving the structure of the underlying image,   the method 
	proposed in this 
	paper, H-COROSA  is always the best performing one.  It can also be 
	observed that   
	the margins in the scores
	yielded by the proposed method  are  more significant in the case of 10\% 
	sample sets compared
	to the case of 20\% sample sets.  This   means that H-COROSA can be a 
	valuable 
	tool when it
	is required to reduce the sampling density significantly. 
	An example can be seen in Figure 
	\ref{fig:mri_deep_t4_random_sr_1_noiseless}, 
	where the best reconstruction results for $T4$ from 10\% noiseless samples 
	are
	shown. The H-COROSA reconstruction is clearly superior to all other methods 
	visually, with less blurring and no spurious artifacts. On the other hand, 
	PNP 
	reconstructions have higher blurring while DAGAN reconstruction shows 
	artificial 
	structures. It may be noted that in terms of scores for this 
	reconstruction, 
	Table \ref{table:deep_noiseless} shows higher SSIM score for H-COROSA, as 
	confirmed by visual inspection. However, DAGAN and PNP have higher SNR 
	scores, 
	in spite of higher blurring and spurious artifacts as noted previously. The 
	scan lines from the selected regions also show that H-COROSA  follows the 
	original intensity profile more consistently and closely, compared to the 
	other 
	methods. Table \ref{table:deep_noisy}  shows that H-COROSA has higher SSIM 
	scores over other methods in all images with both 10\% and 20\% sampling 
	densities, in the presence of measurement noise. As seen in the case of 
	noiseless samples, SNR scores of H-COROSA are insignificantly  lower than 
	competing methods,  but,  in terms of SSIM 
	scores,  H-COROSA is always better as claimed before. Overall, H-COROSA has 
	the 
	advantage of consistent and improved performance over the deep learning 
	methods, while avoiding the requirements of appropriate training sets which 
	are not always available.  
	
	\begin{table}[htbp]
		\centering
		\caption{MRI Reconstruction: Comparison with deep learning methods for 
			noiseless undersampling}
		\renewcommand{\arraystretch}{0.95}
		\begin{tabular}[t]{ccccc}
			\toprule
			\multirow{2}{*}{T1} &\multicolumn{2}{c}{SSIM} 
			&\multicolumn{2}{c}{SNR} 
			\\  &10\% &20\% &10\% &20\%\\
			\midrule
			H-COR. &\textbf{.993} &\textbf{.999} &\textbf{29.51} 
			&\textbf{38.22}\\
			COR. &.991 &\textbf{.999} &28.61 &37.82\\
			DAG. &.961 &.975 &28.65 &31.74\\
			PNP &.970 &.986 &29.18 &33.12\\
			\toprule
			\multirow{2}{*}{T2} &\multicolumn{2}{c}{SSIM} 
			&\multicolumn{2}{c}{SNR} 
			\\  &10\% &20\% &10\% &20\%\\
			\midrule
			H-COR. &\textbf{.992} &\textbf{.998} &\textbf{28.93} 
			&\textbf{37.08}\\
			COR. &.988 &\textbf{.998} &27.21 &36.62\\
			DAG. &.959 &.972 &28.32 &31.19\\
			PNP &.970 &.986 &28.57 &32.88\\
			\toprule
			\multirow{2}{*}{T3} &\multicolumn{2}{c}{SSIM} 
			&\multicolumn{2}{c}{SNR} 
			\\  &10\% &20\% &10\% &20\%\\
			\midrule
			H-COR. &\textbf{.983} &\textbf{.997} &26.37 &\textbf{35.14}\\
			COR. &.978 &\textbf{.997} &25.00 &34.59\\
			DAG. &.951 &.968 &27.41 &29.74\\
			PNP &.972 &.986 &\textbf{27.58} &32.57\\
			\bottomrule
		\end{tabular}
		\renewcommand{\arraystretch}{0.95}
		\begin{tabular}[t]{ccccc}
			\toprule
			\multirow{2}{*}{T4} &\multicolumn{2}{c}{SSIM} 
			&\multicolumn{2}{c}{SNR} 
			\\  &10\% &20\% &10\% &20\%\\
			\midrule
			H-COR. &\textbf{.978} &\textbf{.996} &25.01 &\textbf{33.67}\\
			COR. &.971 &\textbf{.996} &23.14 &32.94\\
			DAG. &.946 &.962 &\textbf{27.20} &29.27\\
			PNP &.973 &.986 &26.61 &32.02\\
			\toprule
			\multirow{2}{*}{T5} &\multicolumn{2}{c}{SSIM} 
			&\multicolumn{2}{c}{SNR} 
			\\  &10\% &20\% &10\% &20\%\\
			\midrule
			H-COR. &\textbf{.970} &\textbf{.995} &23.91 &\textbf{32.66}\\
			COR. &.962 &.994 &22.63 &31.94\\
			DAG. &.931 &.959 &\textbf{25.61} &28.96\\
			PNP &.965 &.985 &25.46 &31.42\\
			\toprule
			\multirow{2}{*}{T6} &\multicolumn{2}{c}{SSIM} 
			&\multicolumn{2}{c}{SNR} 
			\\  &10\% &20\% &10\% &20\%\\
			\midrule
			H-COR. &\textbf{.968} &\textbf{.994} &23.81 &\textbf{31.86}\\
			COR. &.956 &.993 &22.03 &31.05\\
			DAG. &.925 &.955 &\textbf{24.62} &28.48\\
			PNP &.962 &.985 &24.41 &31.14\\
			\bottomrule
		\end{tabular}
		\label{table:deep_noiseless}
	\end{table}
	
	\begin{figure}[htbp]
		\centering
		\includegraphics[width= 
		0.98\textwidth]{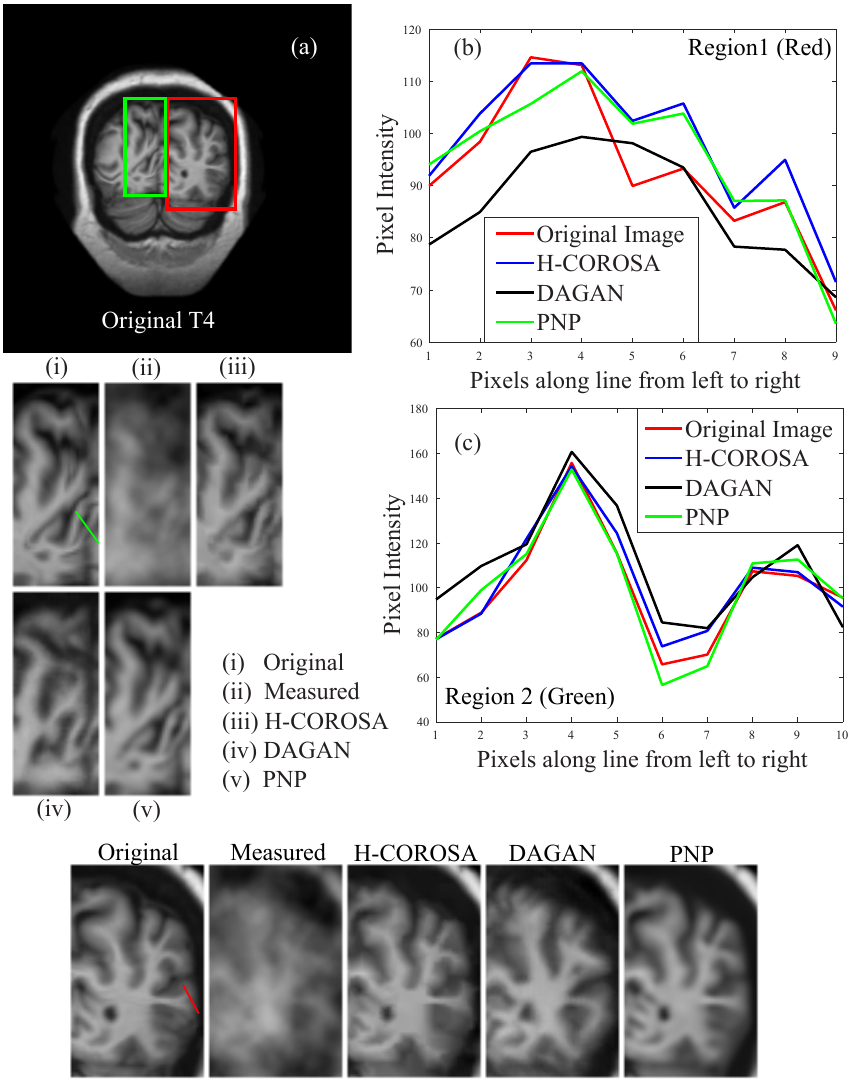}
		\caption{Reconstruction of $T4$ from 10\% random noiseless samples with 
			selected regions, scanlines and corresponding intensity profiles.
			Here `measured'  denotes  zero-filled inverse DFT of
			measured samples.}
		\label{fig:mri_deep_t4_random_sr_1_noiseless}
	\end{figure}
	
	\begin{table}[htbp]
		\centering
		\caption{MRI Reconstruction: Comparison with deep learning methods for 
			noisy samples}
		\renewcommand{\arraystretch}{0.95}
		\begin{tabular}[t]{ccccc}
			\toprule
			\multirow{2}{*}{T1} &\multicolumn{2}{c}{SSIM} 
			&\multicolumn{2}{c}{SNR} 
			\\  &10\% &20\% &10\% &20\%\\
			\midrule
			H-COR. &\textbf{.924} &\textbf{.960} &19.93 &23.81\\
			COR. &.919 &.953 &19.73 &22.77\\
			DAG. &.802 &.872 &18.71 &24.19\\
			PNP &.807 &.915 &\textbf{20.15} &\textbf{24.50}\\
			\toprule
			\multirow{2}{*}{T2} &\multicolumn{2}{c}{SSIM} 
			&\multicolumn{2}{c}{SNR} 
			\\  &10\% &20\% &10\% &20\%\\
			\midrule
			H-COR. &\textbf{.917} &\textbf{.958} &19.62 &23.76\\
			COR. &.906 &.949 &19.26 &22.66\\
			DAG. &.804 &.877 &18.96 &24.15\\
			PNP &.777 &.918 &\textbf{19.87} &\textbf{24.55}\\
			\toprule
			\multirow{2}{*}{T3} &\multicolumn{2}{c}{SSIM} 
			&\multicolumn{2}{c}{SNR} 
			\\  &10\% &20\% &10\% &20\%\\
			\midrule
			H-COR. &\textbf{.915} &\textbf{.942} &19.22 &22.65\\
			COR. &.902 &.935 &18.75 &22.10\\
			DAG. &.783 &.769 &17.89 &21.61\\
			PNP &.783 &.875 &\textbf{19.45} &\textbf{23.44}\\
			\bottomrule
		\end{tabular}
		\renewcommand{\arraystretch}{0.95}
		\begin{tabular}[t]{ccccc}
			\toprule
			\multirow{2}{*}{T4} &\multicolumn{2}{c}{SSIM} 
			&\multicolumn{2}{c}{SNR} 
			\\  &10\% &20\% &10\% &20\%\\
			\midrule
			H-COR. &\textbf{.904} &\textbf{.937} &18.17 &22.02\\
			COR. &.898 &.928 &18.30 &21.59\\
			DAG. &.777 &.666 &17.52 &18.70\\
			PNP &.764 &.862 &\textbf{18.60} &\textbf{22.38}\\
			\toprule
			\multirow{2}{*}{T5} &\multicolumn{2}{c}{SSIM} 
			&\multicolumn{2}{c}{SNR} 
			\\  &10\% &20\% &10\% &20\%\\
			\midrule
			H-COR. &\textbf{.901} &\textbf{.931} &18.46 &21.47\\
			COR. &.887 &.921 &17.91 &21.25\\
			DAG. &.757 &.696 &16.34 &19.42\\
			PNP &.795 &.847 &\textbf{18.68} &\textbf{22.06}\\
			\toprule
			\multirow{2}{*}{T6} &\multicolumn{2}{c}{SSIM} 
			&\multicolumn{2}{c}{SNR} 
			\\  &10\% &20\% &10\% &20\%\\
			\midrule
			H-COR. &\textbf{.891} &\textbf{.934} &18.11 &21.73\\
			COR. &.877 &.922 &17.47 &21.24\\
			DAG. &.728 &.685 &15.27 &18.59\\
			PNP &.748 &.865 &\textbf{18.32} &\textbf{22.31}\\
			\bottomrule
		\end{tabular}
		\label{table:deep_noisy}
	\end{table}
	
	\section{Discussion and Conclusion}

	In any image reconstruction problem,  including the problem of 
	reconstructing
	from Fourier samples,  enforcing a prior model in the characteristics of the
	underlying image is inevitable. The quality of reconstructed images
	is largely determined by the suitability of the prior.  Two classes of image
	reconstruction methods that are dominant in the current literature are 
	the total variation (TV) regularization methods and deep learning methods.  
	TV
	regularization 
	methods rely on hand-crafted prior models that are often proven to be
	ad hoc leading to distortions in the reconstructed images.   Deep learning
	methods  are believed to be less distorting because they enforce natural
	prior derived from training, but, their need for  training data is their 
	disadvantage.
	Besides, concerns about some stability  issues have also been  raised  
		regarding these
	methods \cite{Deep_Instability}.

	The proposed regularization method, H-COROSA,  which is an extension of  
	the TV 
	method,  
	does not give distortions under the conditions of severe under-sampling and 
	severe noise.
	It outperforms all known variants of total variation methods and well-known 
	deep learning
	methods.  H-COROSA can yield good quality reconstruction from highly 
	under-sampled data.  Hence it can be useful to increase the 
	frame rate of MR imaging
	by setting the imaging system to scan   a small fraction in  the Fourier 
	plane.  
	Further,  H-COROSA
	will allow improving frame rate in a wider range of studies compared to 
	deep 
	learning methods,
	because H-COROSA does not need training data. The key factor that led to 
	the 
	success of
	H-COROSA is the spatially adaptive weighting of different types of imaging 
	derivatives.  This leads to a
	reduction in distortions that  normally occur in the non-adaptive 
	regularizations utilizing
	derivatives.  The key enabling factor for this adaptivity is the  
	multiresolution scheme where
	the information extracted from the lower resolution reconstruction helps to 
	improve the quality
	of reconstruction in the next higher resolution level.
	
	\appendix
	
	\section{Overview of Derivative Based Regularization methods}
	
	The earliest form of regularization    used for solving image 
	reconstruction 
	problems is called
	the Tikhonov regularization \cite{Tikhonov}.  Tikhonov regularized 
	reconstruction can be
	expressed as
	\begin{equation}
		\begin{gathered}
			\label{eq:tikh1}
			s_{opt} = \underset{s}{\operatorname{argmin}}\;\;F(s,h,m) + \lambda
			\sum_{\bf r} \|({\bf d}_1*s)({\bf r}) \|_2^2  \\
			\mbox{with}\;\; d_{1}(\r) = [d_{x}(\r), d_{y}(\r)]^{T},
		\end{gathered}
	\end{equation}
	where $d_{x}(\r)$ and $d_{y}(\r)$ are filters implementing first-order 
	derivatives,
	$\frac{\partial}{\partial x}$, and $\frac{\partial}{\partial y}$ 
	respectively.  
	The above computational
	problem is easy to solve,  but, this formulation leads to a severe loss of 
	resolution   under the noise
	levels ranging from moderate to high.   As an improvement, the first-order 
	total variation 
	(TV1) \cite{TV1_Rudin} method was proposed, 
	which can be expressed as
	\begin{equation}
		\begin{gathered}
			\label{eq:tv1}
			s_{opt} = \underset{s}{\operatorname{argmin}}\;\;F(s,h,m) + \lambda
			\sum_{\bf r} \|({\bf d}_1*s)({\bf r}) \|_2 
		\end{gathered}
	\end{equation}
	The removal of squaring in the regularization leads to 
	a better resolution and TV has been widely used 
	\cite{TV_App_BlindDecon, 
		TV_App_Microwave, TV_App_Seismic, TV_App_Wavelet, TV_App_Wlinpaint}
	because of  its ability to recover sharp image features in the
	presence of noise and in the cases of under-sampling. 
	While TV1 is able to retain edges \cite{TV_Edge} in the reconstruction as 
	compared to 
	Tikhonov regularization \cite{Tikhonov}, it presents drawbacks such as
	staircase artifacts  \cite{TV_Ring_Stair, TV_Stair_Jalalzai}.   
	Higher-order extensions of TV  \cite{TV_Chambolle_97, Scherzer_tv2_98, 
		TV_Chan_2000, HS_13} have been proposed to avoid staircase artifacts 
		and 
	they
	deliver better restoration, albeit at the cost of increased computations. 
	Second-order TV (TV2)  \cite{Scherzer_tv2_98} restoration was proposed as 
	\begin{equation}
		\begin{gathered}
			\label{eq:tv2}
			s_{opt} = \underset{s}{\operatorname{argmin}}\;\;F(s,h,m) + \lambda
			\underbrace{\sum_{\bf r}\|({\bf d}_2*s)({\bf r})\|_2}_{R_2(s)}, \\
			\mbox{with}\;\; {\bf d}_2({\bf r}) = [d_{xx}({\bf 
			r})\;\;d_{yy}({\bf 
				r})\;\;\sqrt{2}d_{xy}({\bf r})]^T,
		\end{gathered}
	\end{equation}
	where $d_{xx}(\r), d_{yy}(\r)$, and $d_{xy}(\r)$ are discrete filters 
	implementing second-order derivatives 
	$\frac{\partial^2}{\partial x^2}$, $\frac{\partial^2}{\partial y^2}$ and 
	$\frac{\partial^2}{\partial x\partial y}$ 
	respectively. Another second-order derivative-based formulation is 
	Hessian-Schatten (HS) norm regularization  \cite{HS_13}, which has been 
	proposed as a generalization
	of the standard TV2 regularization.  It is constructed as an $\ell_p$ norm 
	of 
	Eigenvalues of the Hessian matrix, which
	becomes the standard TV2 for $p=2$.  HS norm with $p=1$ has been proven to 
	yield the best resolution
	in the reconstruction  since this better preserves Eigenvalues   of the 
	Hessian  \cite{HS_13}.
	Let  ${\bf H}_2({\bf r})$ be  the matrix filter  composed of   $d_{xx}({\bf 
		r})$, $d_{yy}({\bf r})$, 
	and $d_{xy}({\bf r})$  and let ${\pmb \zeta}(\cdot)$ be the operator that 
	returns the vector containing
	the Eigenvalues of its matrix argument. Then HS norm regularization of 
	order 
	$p$ is expressed as
	\begin{equation}
		\label{eq:hs_reg}
		s_{opt} = \underset{s}{\operatorname{argmin}}\;\;F(s,h,m) + \lambda
		\sum_{\bf r}\left\|{\pmb \zeta}(({\bf H}_2*s)({\bf r}))\right\|_p.
	\end{equation}
	Since the Eigenvalues are actually directional second derivatives taken 
	along 
	principle directions, setting $p=1$
	better preserves the local image structure. It has to be noted that the 
	costs 
	given in the equations  (\ref{eq:tv1}), and (\ref{eq:tv2})   are often 
	minimized using 
	gradient-based approaches with smooth approximations of the form $R_k(s) = 
	{\sum \limits_{\r}} \sqrt{\epsilon + 
		||d_{k}(\r) \ast s(\r)||_{2}^2}, \ k=1,2$ where $\epsilon$ is a small 
	positive 
	constant   \cite{TV1_Rudin, vogel_tv1, TV_Smooth_bect20041, 
		TV_Smooth_bioucas2007new}.
	This approach  has been proven to converge to the minimum of the exact form 
	as  
	$\epsilon\rightarrow 0$
	\cite{vogel_tv1}.  Approaches to minimize the cost without smooth 
	approximation include the primal-dual
	method  \cite{TV_Dual_Chambolle},  and  alternating direction of multiplier 
	method (ADMM)   \cite{tvadmm, TV_ADMM_Afonso, 
		TV_ADMM_Steidl, TV_ADMM_Manu}.
	A detailed comparison of such approaches  has been provided in  
	\cite{TV_Aujol}. 
	
	It has been demonstrated that combining first- and second-order
	derivatives is advantageous in accurately restoring image features  
	\cite{TV_CombinedLysaker, Combined_order_TV, TGV2, sam-tv}. 
	In this regard, the combined order TV  \cite{Combined_order_TV} 
	uses scalar relative weights for combining first- and second-order 
	variations, 
	with the relative weights left as user parameters, and the solution is 
	estimated 
	by means of the optimization problem of
	the form 
	\begin{equation}
		s_{opt} = \underset{s}{\operatorname{argmin}}\;\;F(s,h,m)
		+   \alpha_1 R_1(s) +   \alpha_2 R_2(s),
	\end{equation} 
	where $\alpha_1$  and $\alpha_2$ determine the relative weight.
	A generalization for total variation to higher-order terms, named  
	total generalized variation (TGV)  has also been proposed \cite{TGV, 
	TGV2}.  
	It is 
	generalized in the following ways:  it is formulated for any
	general derivative order,  and for any given order,  it is generalized in 
	the 
	way
	how the derivatives are penalized.  Only the second-order TGV form 
	\cite{TGV2} 
	has been
	well explored for image reconstruction, which takes the following form: 
	\begin{equation}
		(s_{opt}, {\bf p}_{opt})  = \underset{s, {\bf 
				p}}{\operatorname{argmin}}\;\;F(s,h,m)  + 
		\alpha_1\sum_{\bf r} \|({\bf d}_1*s)({\bf r})-{\bf p}({\bf r})\|_2 
		+ \alpha_2 
		\frac{1}{2} \sum_{\bf r}  \|{\bf d}_1({\bf r})*{\bf p}^T({\bf r})+{\bf 
			p}({\bf r})*{\bf d}_1^T({\bf r})\|_F,
	\end{equation}
	where ${\bf p}({\bf r})$ is an auxiliary $2\times 1$ vector image.
	The TGV functional is able to spatially adapt to the underlying image 
	structure 
	because of the minimization w.r.t.
	auxiliary variable
	${\bf p}$. Near edges, ${\bf p}({\bf r})$ approaches zero leading to 
	TV1-like 
	behavior which allows sharp jumps in
	the edges.  On the other hand,  in smooth regions,  ${\bf p}({\bf r})$ 
	approaches  ${\bf d}_1*s({\bf r})$ leading
	to TV2-like behavior which will avoid staircase artifacts. However, the 
	drawback with TGV
	functional is that   the weights  $\alpha_1$  and $\alpha_2$ have to be 
	chosen 
	by the user.

	\section{Computation of discrete derivatives}
	
	For discrete image $s(\r)$,  discrete implementation  of derivative 
	operators
	$\frac{\partial}{\partial x}$,  $\frac{\partial}{\partial y}$,  
	$\frac{\partial^2}{\partial x^2}$,
	$ \frac{\partial^2}{\partial y^2}$, and  $\frac{\partial^2}{\partial 
	x\partial 
		y}$
	are equivalent to filtering by the filters  $d_x$,  $d_y$,  $d_{xx}$, 
	$d_{yy}$, 
	and
	$d_{xy}$.  With $\r = (x,y)$,
	filtering with  these filters can be expressed as:
	\begin{align}
		(d_x*s)(\r)  & = s(x,y) -  s(x-1,y),  \\
		(d_y*s)(\r)  & = s(x,y) -  s(x,y-1),  \\
		(d_{xx}*s)(\r) & =  -2s(x,y) +   s(x-1,y) - s(x+1,y),   \\
		(d_{yy}*s)(\r)  & =  -2s(x,y) +   s(x,y-1) - s(x,y+1),   \\
		(d_{xy}*s)(\r)  & =   s(x,y) + s(x-1,y-1) -  s(x-1,y) -  s(x,y-1),
	\end{align}
	
	\section{Regularization for adaptive weight computation}
	
	Given a guide image,  $\hat{s}$,   the required adaptive weights are 
	supposed 
	to be the
	minimizer of  $R(\hat{s},{\beta}_{1},  {\beta}_{2},  {\beta}_{2})$  with 
	respect to weight
	images  ${\beta}_{1}$,  ${\beta}_{2}$,  and ${\beta}_{2}$.  Recall that 
	$R(\hat{s},{\beta}_{1},  {\beta}_{2},  {\beta}_{2})$   is given by 
	\begin{equation}
		R(\hat{s},{\beta}_{1},  {\beta}_{2}, 
		{\beta}_{2})   = \sum_{\r} \left(  
		{\beta}_1(\r)\|\nabla \hat{s}({\bf r})\|_2  +   {\beta}_2(\r) |{\cal 
			H}^+\hat{s}({\bf r})|   
		+  {\beta}_3(\r) |{\cal H}^-\hat{s}({\bf r})|
		\right) 
	\end{equation}
	Minimizing $R(\hat{s},{\beta}_{1},  {\beta}_{2},  {\beta}_{2})$  w.r.t.  
	${\beta}_{1}$,  ${\beta}_{2}$,  and ${\beta}_{2}$
	subject to  the conditions $0\le \beta_i(\r) \le 1$ and $\sum_{i=1}^3 
	\beta_i(\r) =1$ will give a solution in the following
	form:  for each $\r$,  one of $\beta_i(\r)$s will be $1$ and others will be 
	zeros.  The value of $i$ for which
	$\beta_i(\r)$  is $1$ will be the index of the element in the list 
	$\{\|\nabla 
	\hat{s}({\bf r})\|_2, 
	|{\cal H}^+\hat{s}({\bf r})|,   |{\cal H}^-\hat{s}({\bf r})|\}$ that has 
	the 
	lowest value.  This means the
	adaptive weights obtained as the minimum of $R(\hat{s},{\beta}_{1},  
	{\beta}_{2},  {\beta}_{2})$  will be rapidly
	changing between 0 and 1  as move along the pixel indices to cope up with
	even insignificant differences  among the terms
	$\{\|\nabla  \hat{s}({\bf r})\|_2$,    $|{\cal H}^+\hat{s}({\bf r})|$,  
	and   $|{\cal H}^-\hat{s}({\bf r})|\}$.
	This will lead to 
	artifacts  in the 
	reconstruction.  Due to this reason,  we determine the adaptive weights by 
	minimizing $R(\hat{s},{\beta}_{1},  {\beta}_{2}, {\beta}_{2}) + 
	L(\beta_1,\beta_2,\beta_2, \tau)$
	where $L(\beta_1,\beta_2,\beta_3,\tau) = -\tau\sum_\r 
	log(\beta_1(\r)\beta_2(\r)\beta_3(\r))$.  
	We use the term $L(\beta_1,\beta_2, \beta_3, \tau)$ to  prevent   spurious 
	changes in 
	the values of the weights 
	$\beta_1(\r)$,  $\beta_2(\r)$,  and   $\beta_3(\r)$.
	Any of the  weights $\beta_1(\r)$,  $\beta_2(\r)$,  and   $\beta_3(\r)$ 
	will approach zero
	only if the corresponding  entry  in  the list is 
	$\{\|\nabla 
	\hat{s}({\bf r})\|_2, 
	|{\cal H}^+\hat{s}({\bf r})|,   |{\cal H}^-\hat{s}({\bf r})|\}$ is 
	significantly lower   than
	the other entries.  Further,  if any of the  weights $\beta_1(\r)$,  
	$\beta_2(\r)$,  and   $\beta_3(\r)$ 
	becomes zero, the penalty  goes to $\infty$. This ensures that all terms 
	are 
	involved at every image pixel  leading to robustness against noise.
	The role of $\tau$ is to control the  sensitivity of variations in the 
	weight  
	$\beta_1(\r)$,  $\beta_2(\r)$,  and   $\beta_3(\r)$   in response to the 
	variations in the terms
	$\|\nabla 
	{s}({\bf r})\|_2$, $ {\cal H}^+s({\bf r})$, and ${\cal H}^-s({\bf r})$. A 
	low 
	value of $\tau$ imposes less penalty through $L(\beta_1,\beta_2, \beta_3, 
	\tau)$ and hence, will promote a rapid  variation in the weights 
	$\beta_1(\r)$,  $\beta_2(\r)$,  and   $\beta_3(\r)$. Further, it will also 
	promote
	the weights going close to zero.
	On the 
	other hand, a high value of $\tau$ will prevent  
	rapid  variation in the weights 
	$\beta_1(\r)$,  $\beta_2(\r)$,  and   $\beta_3(\r)$;   further,  it  will 
	also prevent the weights
	from going close to zero,   and 
	thus will make  the regularization  to use all three terms in a more 
	balanced way.  
	Hence, if $\tau$ is too high,  the structural adaptivity of the 
	regularization will be lost.
	This means that a balanced choice of  value  for 
	$\tau$ crucial.  Fortunately, we observed
	experimentally that a nominal value for $\tau$   that works well 
	for various types of images and noise levels can be chosen.

	\section{Notations and conventions for algorithm description}
	\label{sec:notcon}
	
	For a 2D discrete filter,  $g(\r)$ where $\r \in {\mathbb N}^2$,  $g^T(\r)$ 
	denotes the flipped
	filter, i.e.,  $g^T(\r)=g(-\r)$.  The discrete convolution of the filter 
	$g(\r)$ with an image $x(\r)$, 
	denoted by $(g*x)(\r)$, is given by $(g*x)(\r) = \sum_{\t}x(\t)g(\r-\t)$.  
	With 
	this,  the
	flipped convolution will be $(g^T*x)(\r) = \sum_{\t}x(\t)g(\r+\t)$.  For a 
	vector discrete filter,
	$\g(\r) = [g_1(\r)\;g_2(\r)]^T$,  where each $g_i$ is a discrete filter,  
	and 
	for an image $x(\r)$,
	the notation $(\g*x)(\r)$ denotes $[(g_1*x)(\r)\;(g_2*x)(\r)]^T$, which is 
	a 
	vector image.  For a
	vector image,  $\y(\r) = [y_1(\r)\;y_2(\r)]^T$,  and a vector filter,  
	$\g(\r) 
	= [g_1(\r)\;g_2(\r)]^T$, 
	the notation $(\g^T*\y)(\r)$  denotes  $(g^T_1*y_1)(\r) + 
	(g^T_2*y_2)(\r)$.  
	For scalar images,
	$x(\r)$ and $y(\r)$,  $(xy)(\r)$ denotes the point-wise product, i.e.,  
	$(xy)(\r)=x(\r)y(\r)$.
	For a vector image  $\y(\r) = [y_1(\r)\;y_2(\r)]^T$ and a scalar image 
	$z(\r)$,  the Kronocker 
	product,  $(z\oplus \y)(\r)$, is  given by   $(z\oplus \y)(\r) =   
	[(zy_1)(\r)\;(zy_2)(\r)]^T$. 
	For $3\times 1$ vector filters and vector images,  these definitions are 
	extended in a similar way.

	For a pair of images  $x(\r)$  and $y(\r)$,  we define the inner product as
	$\langle x,y \rangle = \sum_{\r} x(\r)y(\r).$
	The norm $\|x\|_2$ is given by
	$\|x\|_2= \sqrt{\langle x,x \rangle}$.  For vector images,  
	$\x(\r)$  and $\y(\r)$, the inner product, $\langle \x(\r), \y(\r) \rangle$
	denotes  $\x^T(\r)\y(\r)$,   and $\|\x(\r)\|_{2}$ denotes 
	$\sqrt{\x^T(\r)\x(\r)}$.
	On the other hand,  writing the same expressions without the pixel index
	will indicate the summation over pixel indices.  In other words, we
	have that $\langle \x, \y \rangle  = \sum_{\r} \x^T(\r)\y(\r),$  and 
	$\|\x\|_2= \sqrt{\langle \x,\x \rangle}$.   
	The notation $\|\x\|_{\cdot2}$  denotes a scalar  image   that is a 
	point-wise 
	norm
	of the vector image $\x(\r)$, i.e.,  
	$\|\x\|_{\cdot2}(\r) =  \sqrt{\x^T(\r)\x(\r)}$.   With this, 
	$\langle x, \|\y\|_{\cdot2}\rangle$ will denote the sum 
	$ \sum_{\r} x(\r)\|\y\|_{\cdot2}(\r)$. 
	
	We will   use  operators denoting the computation of Eigenvalues of 
	a $2\times 2$
	symmetric  matrices.    Let $\mathbb{S}^{2\times 2}$ denote the space of 
	$2\times 2$
	symmetric  matrices.  Let ${\bf v}= \genm{v_1}{v_2}{v_3}$.   We will use 
	the 
	notations
	${\pmb \Lambda}_1({\bf v})$ and ${\pmb \Lambda}_2({\bf v})$  for denoting 
	the 
	Eigenvalues
	$ {\bf v}$.  In other words, 
	$${\pmb \Lambda}_1({\bf v}) = 
	0.5\left(v_1+v_2+\sqrt{(v_1-v_2)^2+4v_3}\right)$$
	$${\pmb \Lambda}_2({\bf v}) = 
	0.5\left(v_1+v_2-\sqrt{(v_1-v_2)^2+4v_3}\right)$$
	For a $2\times 2$ symmetric matrix of images,   $\y$,  ${\pmb 
	\Lambda}_1\y$  
	and
	${\pmb \Lambda}_2\y$ denote the images of such Eigenvalues.
	For a  $2\times 2$ symmetric  matrix ${\bf v}$,  its Eigenvectors will be 
	of 
	the form 
	$[\cos(\theta)\;\sin(\theta)]^T$,  
	and  $[-\sin(\theta)\;\cos(\theta)]^T$, where $\theta$ is the angle that 
	depends on 
	${\bf v}$.
	Let $\Theta({\bf v})$ denote the operator that returns this angle.  Let
	${\cal E}(\cdot,\cdot,\cdot)$ denote the operator that reconstructs the 
	$2\times 2$
	symmetric matrix from the Eigenvalues and the angle determining the 
	Eigenvectors
	such that
	${\bf v} = {\cal E}({\pmb \Lambda}_1({\bf v}), {\pmb \Lambda}_2({\bf v}), 
	\Theta({\bf v}))$. 
	
	\section{Adaptive regularized reconstruction with given spatial weights}
	\label{sec:adreg}

	The goal here is to develop the numerical algorithm for the minimizations
	represented by  ${\mathlarger{\bar{\cal M}_{ad}}^{(j)}}(\cdot)$ and   
	${\mathlarger{{\cal M}_{ad}}}(\cdot)$ in the equations
	\eqref{eq:recph12mrloop}  and \eqref{eq:recph12iter}.   Note that  
	${\mathlarger{{\cal M}_{ad}}}(\cdot)$ is a special case of 
	${\mathlarger{\bar{\cal 
				M}_{ad}}^{(j)}}(\cdot)$.
	Specifically, we have ${\mathlarger{\bar{\cal 
				M}_{ad}}^{(0)}}(\cdot)={\mathlarger{{\cal M}_{ad}}}(\cdot)$.
	Hence, it is sufficient to consider  ${\mathlarger{\bar{\cal 
				M}_{ad}}^{(j)}}(\cdot)$. 
	The operation represented by  ${\mathlarger{\bar{\cal 
	M}_{ad}}^{(j)}}(\cdot)$
	is equivalent to solving the following
	minimization problem:
	\begin{equation}
		\label{eq:reccons}
		s^{(j)}   = \underset{\makecell{s=E^{(j)}\bar{s}\\ \bar{s} \in 
				\mathbb{R}^{\frac{N}{2^j}\times \frac{N}{2^j}}} 
				}{\operatorname{argmin}} \;
		F(s,{\bf m}, T)  + 
		\lambda R(s,\hat{\beta}_{1},  \hat{\beta}_{2}, 
		\hat{\beta}_{2})  + {\cal B}(s)
	\end{equation}
	where 
	\begin{align}
		R(s,\hat{\beta}_{1},  \hat{\beta}_{2}, 
		\hat{\beta}_{2})  & = \sum_{\r} \left(  
		\hat{\beta}_1(\r)\|\nabla s({\bf r})\|_2  + \hat{\beta}_2(\r) |{\cal 
			H}^+s({\bf r})|   +  \hat{\beta}_3(\r) |{\cal H}^-s({\bf r})|
		\right),  \\
		F(s,{\bf m}, T)  & = \| Ts-{\bf m} \|_2^2,
	\end{align}
	and  ${\cal B}(s)$ is box constraint function which takes $\infty$ if any 
	of the pixels of $s$ has a value outside
	a given range,  and  takes zero otherwise.  Here the subscript lines under 
	the 
	argmin,  
	$s=E^{(j)}\bar{s}$ and $\bar{s} \in \mathbb{R}^{\frac{N}{2^j}\times 
		\frac{N}{2^j}}$  means that
	the optimization variable $s$ is constrained to be of the form 
	$s=E^{(j)}\bar{s}$ where $E^{(j)}$ is a $2^j$-fold
	interpolator and $\bar{s}$ is an image of size $\frac{N}{2^j}\times 
	\frac{N}{2^j}$.  Here the operator
	${\cal H}^+s({\bf r})$ denotes computation of discrete  derivatives 
	$\frac{\partial^2}{\partial x^2}$, $\frac{\partial^2}{\partial y^2}$ and 
	$\frac{\partial^2}{\partial x\partial y}$ 
	and then the computation of one of the Eigenvalues of the Hessian.    
	Similarly, ${\cal H}^-s({\bf r})$  denotes
	the computation of  the same set of derivatives and then the computation of 
	the 
	other Eigenvalue.  To facilitate the
	development of   the minimization algorithm,  we have to split these 
	notations into 
	two stages.  To this end, let
	$d_{xx}(\r), d_{yy}(\r)$, and $d_{xy}(\r)$ be  discrete filters implementing
	second-order derivatives  $\frac{\partial^2}{\partial x^2}$, 
	$\frac{\partial^2}
	{\partial y^2}$ and $\frac{\partial^2}{\partial x\partial y}$  
	respectively, 
	and 
	let ${\bf h}({\bf r}) = \genm{d_{xx}({\bf r})}{d_{yy}({\bf r})}{d_{xy}({\bf 
			r})}$.  
	Then $({\bf h}*s)(\r)$ is $2\times 2$ matrix of images. 
	Then we have
	$${\cal H}^+s({\bf r}) =   {\pmb \Lambda}_1(({\bf h}*s)(\r)),$$
	$${\cal H}^-s({\bf r}) =   {\pmb \Lambda}_2(({\bf h}*s)(\r)),$$
	where ${\pmb \Lambda}_1(\cdot)$  and ${\pmb \Lambda}_2(\cdot)$  are the 
	operators
	defined in Section \ref{sec:notcon}.
	Next,  let
	${\bf g}(\r) = [d_{x}(\r), d_{y}(\r)]^{T}$
	where $d_{x}(\r)$ and $d_{y}(\r)$ are filters implementing first order 
	derivatives,
	$\frac{\partial}{\partial x}$, and $\frac{\partial}{\partial y}$ 
	respectively.  
	Then the adaptive
	regularization,  $R(s,\hat{\beta}_{1},  \hat{\beta}_{2}, \hat{\beta}_{3})$ 
	can 
	be expressed
	as
	\begin{equation}
		R(s,\hat{\beta}_{1},  \hat{\beta}_{2}, 
		\hat{\beta}_{3})  = \sum_{\r} \left(  
		\hat{\beta}_1(\r)\|({\bf g}*s)({\bf r})\|_2  + \hat{\beta}_2(\r) |{\pmb 
			\Lambda}_1(({\bf h}*s)(\r))|   
		+  \hat{\beta}_3(\r) |{\pmb \Lambda}_2(({\bf h}*s)(\r))| \right).
	\end{equation}
	Next, the constrained optimization problem of equation (\ref{eq:reccons}) 
	can 
	be written
	in unconstrained form by absorbing the interpolator $E^{(j)}$ into the cost 
	function.  Specifically,
	the cost can be written as
	\begin{equation}
		\label{eq:recuc}
		\bar{s}^{(j)}   = \underset{\makecell{ \bar{s} \in 
				\mathbb{R}^{\frac{N}{2^j}\times \frac{N}{2^j}}} 
				}{\operatorname{argmin}} \;
		F(E^{(j)}\bar{s},{\bf m}, T) + 
		\lambda R(E^{(j)}\bar{s},\hat{\beta}_{1},  \hat{\beta}_{2}, 
		\hat{\beta}_{2})  + {\cal B}(E^{(j)}\bar{s}).
	\end{equation}
	Here the notations $F(E^{(j)}\bar{s},{\bf m}, T)$,   
	$R(E^{(j)}\bar{s},\hat{\beta}_{1},  \hat{\beta}_{2}, \hat{\beta}_{3})$,
	and ${\cal B}(E^{(j)}\bar{s})$ have straightforward interpretation.  They 
	represent functions obtained
	by replacing  variable $s$ with the $2^j$-fold interpolated image  
	$E^{(j)}\bar{s}$.  We will also use the
	notations  $F^{(j)}(\bar{s},{\bf m}, T)$,   
	$R^{(j)}(\bar{s},\hat{\beta}_{1},  
	\hat{\beta}_{2}, \hat{\beta}_{3})$,
	and ${\cal B}^{(j)}(\bar{s})$ to represent these functions.    Note that, 
	when 
	$j=0$,  there is no interpolation,
	and hence $E^{(0)}$ denote the identity operator (the operator that returns 
	input as the output). 
	This means that  variable $\bar{s}$ is of original size ($N\times N$).   
	Let $\mathlarger{{\cal M}_{ad}}^{(j)}\big(\cdot\big)$  denote the operator 
	that
	is equivalent to the modified minimization given in the equation 
	\eqref{eq:recuc}.
	It can be written as 
	\begin{align}
		\label{eq:recph1mr}
		\bar{s}^{(j)} =   &\mathlarger{{\cal M}_{ad}}^{(j)}\Big(
		\scalebox{.7}[1.0]{\bf OptInit}
		=\hat{s},  \;\ 
		\scalebox{.7}[1.0]{\bf AdapWt} = \{\hat{\beta}_{1},  
		\hat{\beta}_{2}, 
		\hat{\beta}_{3}\}, \;\; 
		\scalebox{.7}[1.0]{\bf Data}=\{T, {\bf m}\},  \;\;
		\lambda \Big).
	\end{align}
	Note that the reconstruction problems of equations  (\ref{eq:reccons})  
	and  
	(\ref{eq:recuc}) and nearly
	identical with a minor notational difference:  the minimization of 
	equation  
	(\ref{eq:reccons})  searches  for
	image $s$ of size $N\times N$ with a constraint that  $s=E^{(j)}\bar{s}$  
	where  $\bar{s}$ is an image of size
	$\frac{N}{2^j}\times \frac{N}{2^j}$.  On the other hand, the minimization 
	of 
	equation  (\ref{eq:recuc})
	directly search for the reduced image, $\bar{s}$.
	The multiresolution loop of equation  (\ref{eq:recph12mrloop})  can be 
	re-expressed using the modified operator  $\mathlarger{{\cal 
	M}_{ad}}^{(j)}\big(\cdot\big)$
	as given below:
	\begin{equation}
		\label{eq:recph12mrloopred}
		\bar{s}^{(j)} =   \mathlarger{{\cal M}_{ad}^{(j)}}\Big(
		\scalebox{.7}[1.0]{\bf optInit}
		=E^{(1)}{s}^{(j+1)},  \;\ 
		\scalebox{.7}[1.0]{\bf AdaptiveWeights} = 
		\mathlarger{{\cal M}_w}\Big(
		\scalebox{.7}[1.0]{\bf AdaptGuide}= E^{(j+1)}{s}^{(j+1)},   \tau \Big),
		\;\;  
		\scalebox{.7}[1.0]{\bf Data}=\{T, {\bf m}\},  \;\;
		\lambda \Big)
	\end{equation}
	Since the size of the minimization variable changes with the index $j$ in modified scheme,  
	i.e.,  since the
	size of the variable is
	$\frac{N}{2^j}\times \frac{N}{2^j}$, we have two changes in the 
	multiresolution loop
	given above:
	first, the result of the previous reconstruction in 
	the loop,  ${s}^{(j+1)}$
	goes through  the $2$-fold interpolation to generate initialization for 
	numerical optimization; second,  to generate adaptation 
	guide of size $N\times N$,
	the result of previous reconstruction  goes through a   $2^{j+1}$-fold 
	interpolation. 
	
	The goal in the remainder of this section is to develop an algorithm for
	the minimization problem of equation  \eqref{eq:recuc}.  We  use the 
	well-known 
	framework called  the Alternating Direction Method of Multiplier (ADMM).
	Although the ADMM framework is well-known and frequently used for
	image reconstruction problems,  developing an algorithm for
	the minimization problem of equation  \eqref{eq:recuc} using this framework
	is not a routine application and required some innovations.  Our innovations
	are mainly in Sections \ref{sec:consform}  and \ref{sec:zprob}.
	
	\subsection{Transformation of adaptive weights and constrained formulation}
	\label{sec:consform}

	The first step in developing an ADMM algorithm is to re-express the 
	optimization
	problem by introducing auxiliary variables related to the main variable 
	$\bar{s}$
	through a  linear equation such that each sub-cost in the problem of
	equation  \eqref{eq:recuc}  can be minimized independently.  This is called 
	variable splitting.  For our problem, if we extend the splitting scheme 
	that 
	has been used before
	for image reconstruction in a straightforward way, we will get the following
	modified problem:
	\begin{equation}
		\bar{s}^{(j)}   = \underset{\bar{s},x,\y,\z}{\operatorname{argmin}} \;
		F^{(j)}(\bar{s},{\bf m}, T) + 
		\lambda \langle \gamma, \|\y\|_{.2} \rangle
		+  \lambda\innerprod{(1-\gamma)\zeta }{|{\pmb \Lambda}_1 (\z)|} +
		\lambda \innerprod{(1-\gamma)(1-\zeta)}{|{\pmb \Lambda}_2 (\z)|}  + 
		{\cal B}(x)
	\end{equation}
	subject to  $E^{(j)}\bar{s}=x$,  ${\g}*(E^{(j)}\bar{s})=\y$, and
	$\h*(E^{(j)}\bar{s})=\z$.
	However, we observed experimentally that the algorithm derived from the 
	above 
	splitting leads to artifacts in the reconstruction.  We propose an 
	alternative
	variable splitting scheme, which is one of the main contributions of this 
	section.

	Based on our experimental observations,  we present a transformation for 
	the 
	adaptive
	weights  $\{\hat{\beta}_{1},  \hat{\beta}_{2}, \hat{\beta}_{3}\}$  that 
	leads 
	to faster  completion 
	of numerical optimization.  
	Define
	\begin{align}
		\zeta  & = \frac{\hat{\beta}_2}{\hat{\beta}_2+\hat{\beta}_3}\\
		\gamma & =  1 - \hat{\beta}_2 - \hat{\beta}_3
	\end{align}
	With these,  $\hat{\beta}_i$'s can be expressed
	as  
	\begin{align}
		\hat{\beta}_1 & = \gamma \\
		\hat{\beta}_2 &  = (1-\gamma)\zeta \\
		\hat{\beta}_3 & = (1-\gamma)(1-\zeta)
	\end{align}
	Now the constraints $0\le \hat{\beta}_i \le 1$  and 
	$\hat{\beta}_1+\hat{\beta}_2+\hat{\beta}_3=1$ become
	equivalent to $0\le \gamma \le 1$ and $0\le \zeta \le 1$. 
	Now consider
	\begin{equation}
		\resizebox{\hsize}{!}{
			$R^{(j)}(\bar{s},\hat{\beta}_{1},  \hat{\beta}_{2}, 
			\hat{\beta}_{3})  = \sum_{\r} \left(  
			\hat{\beta}_1(\r)\|({\bf g}*(E^{(j)}\bar{s}))({\bf r})\|_2  + 
			\hat{\beta}_2(\r) |{\pmb \Lambda}_1(({\bf 
			h}*(E^{(j)}\bar{s}))(\r))|   
			+  \hat{\beta}_3(\r) |{\pmb \Lambda}_2(({\bf 
			h}*(E^{(j)}\bar{s}))(\r))| 
			\right)$
		}
	\end{equation}
	Using notations introduced in Section \ref{sec:notcon}, the above cost can 
	be 
	written as
	\begin{equation}
		R^{(j)}(\bar{s},\hat{\beta}_{1},  \hat{\beta}_{2}, 
		\hat{\beta}_{3})   = \innerprod{\beta_1}{\pwn{{\g}*(E^{(j)}\bar{s}))}}
		+ \innerprod{\beta_2}{{\pmb \Lambda}_1(\h*(E^{(j)}\bar{s}))} +
		\innerprod{\beta_3}{{\pmb \Lambda}_2 (\h*(E^{(j)}\bar{s}))}
	\end{equation}
	This can be written in terms of the transformed adaptive weights  as
	\begin{equation}
		R^{(j)}(\bar{s},\gamma,\zeta)  = 
		\innerprod{\gamma}{\pwn{{\g}*(E^{(j)}\bar{s}))}}
		+ \innerprod{\zeta(1-\gamma)}{{\pmb \Lambda}_1(\h*(E^{(j)}\bar{s}))} +
		\innerprod{(1-\zeta)(1-\gamma)}{{\pmb \Lambda}_2 (\h*(E^{(j)}\bar{s}))}
	\end{equation}
	This again can re-written as
	\begin{equation}
		R^{(j)}(\bar{s},\gamma,\zeta)  = 
		\pwn{{\gamma}\oplus({\g}*(E^{(j)}\bar{s})))}
		+ \innerprod{\zeta }{{\pmb \Lambda}_1 ((1- 
			\gamma)\oplus(\h*(E^{(j)}\bar{s})))} +
		\innerprod{(1-\zeta)}{{\pmb \Lambda}_2 ((1- 
			\gamma)\oplus(\h*(E^{(j)}\bar{s})))}
	\end{equation}
	Now the minimization problem with updated notations can be expressed as 
	\begin{equation}
		\begin{split}
			\bar{s}^{(j)}    = &  \underset{s}{\operatorname{argmin}} \;
			F^{(j)}(\bar{s},{\bf m}, T)
			+ 
			\lambda \|{\gamma}\oplus({\g}*(E^{(j)}\bar{s})))\|_2 \\
			& +  \lambda\innerprod{\zeta }{{\pmb \Lambda}_1 ((1- 
				\gamma)\oplus(\h*(E^{(j)}\bar{s})))} +
			\lambda \innerprod{(1-\zeta)}{{\pmb \Lambda}_2 ((1- 
				\gamma)\oplus(\h*(E^{(j)}\bar{s})))}  + {\cal 
				B}(E^{(j)}\bar{s})).
		\end{split}
	\end{equation}
	The above minimization can be written as the following 
	constrained  problem, 
	\begin{equation}
		\bar{s}^{(j)}   = \underset{\bar{s},x,\y,\z}{\operatorname{argmin}} \;
		F^{(j)}(s,{\bf m}, T) + 
		\lambda \|\y\|_2
		+  \lambda\innerprod{\zeta }{{\pmb \Lambda}_1 (\z)} +
		\lambda \innerprod{(1-\zeta)}{{\pmb \Lambda}_2 (\z)}  + {\cal B}(x)
	\end{equation}
	subject to  $E^{(j)}\bar{s}=x$,  $\gamma\oplus({\g}*(E^{(j)}\bar{s}))=\y$, 
	and
	$(1- \gamma)\oplus(\h*(E^{(j)}\bar{s}))=\z$.
	
	\subsection{ADMM iterations}
	
	The constrained formulation of reconstruction problem is the key for 
	developing 
	ADMM algorithm.
	The reconstruction problem is now the minimization of following cost,
	\begin{equation}
		\label{eq:isocost}
		{\cal J}_a(\bar{s},x,\y,\z, \lambda) = F^{(j)}(\bar{s},{\bf m}, T) + 
		\lambda \|\y\|_2
		+  \lambda\innerprod{\zeta }{{\pmb \Lambda}_1 (\z)} +
		\lambda \innerprod{(1-\zeta)}{{\pmb \Lambda}_2(\z)}  + {\cal B}(x),
	\end{equation}
	subject to $E^{(j)}\bar{s}=x$,  $\gamma\oplus({\g}*(E^{(j)}\bar{s}))=\y$, 
	and
	$(1- \gamma)\oplus(\h*(E^{(j)}\bar{s}))=\z$.
	The ADMM algorithm is derived through the so-called augmented Lagrangian 
	function \cite{Bertsekas_NLP, admm_eckstein}.  To write
	augmented Lagrangian,  first define
	\begin{equation}
		\label{eq:augmentcost}
		\begin{split}
			C(\bar{s},x,\y,\z,\xh,\yh,\zh)    
			& =     \frac{c}{2} \|\gamma\oplus({\bf 
				g}*(E^{(j)}\bar{s}))-\y\|_2^2    
			+  \innerprod{\hat{\y}}{\gamma\oplus({\bf g}*(E^{(j)}\bar{s}))-\y} 
			\\
			& + c\frac{c}{2} \|(1- \gamma)\oplus({\bf 
				h}*(E^{(j)}\bar{s}))-\z\|_2^2    
			+  \innerprod{\hat{\z}}{(1- \gamma)\oplus({\bf 
			h}*(E^{(j)}\bar{s}))-\z} 
			\\
			& +  \frac{c}{2}  \|E^{(j)}\bar{s}-x\|_2^2    +  
			\innerprod{\hat{x}}{E^{(j)}\bar{s}-x},
		\end{split}
	\end{equation}
	Here $\{\xh,\yh,\zh\}$ are called the Langragian weights.
	Then the following sum is called the augmented Lagrangian:
	\begin{equation}
		L_a(\bar{s},x,\y,\z,\xh,\yh,\zh, \lambda) = {\cal J}_a(\bar{s},x,\y,\z, 
		\lambda) + C(\bar{s},x,\y,\z,\xh,\yh,\zh)    
	\end{equation}
	With this definition,   the ADMM method involves series of
	minimizations on  $L_a(\bar{s},x,\y,\z,\xh,\yh,\zh, \lambda)$, where each 
	minimization is done with respect to one of the
	variables in the set  $\{\bar{s}, x,\y,\z \}$.  Selection of the variable
	for minimization   cycles through the list  $\{\bar{s}, x,\y,\z \}$,
	and each cycle is considered as one step  of the ADMM iteration.  
	In other words, if $k$ is the iteration index,  the update from the current 
	iterate denoted by
	$\{\bar{s}^{(k)}, x^{(k)},\y^{(k)},\z^{(k)}\}$,
	to the next iterate, 
	$\{\bar{s}^{(k+1)}, x^{(k+1)},\y^{(k+1)},\z^{(k+1)}\}$ is done by the 
	following 
	set of minimizations
	and updates:
	\begin{align}
		\y^{(k+1)} &  = \underset{\y}{\operatorname{argmin}} \;
		L_a(\bar{s}^{(k)}, x^{(k)}, \y,\z^{(k)}, \xh^{(k)},\yh^{(k)},\zh^{(k)}, 
		\lambda)  \\
		\z^{(k+1)} &  = \underset{\z}{\operatorname{argmin}} \;
		L_a(\bar{s}^{(k)}, x^{(k)}, \y^{(k+1)},\z, 
		\xh^{(k)},\yh^{(k)},\zh^{(k)},
		\lambda)  \\
		x^{(k+1)} &  = \underset{x}{\operatorname{argmin}} \;
		L_a(\bar{s}^{(k)}, x, \y^{(k+1)},\z^{(k+1)}, 
		\xh^{(k)},\yh^{(k)},\zh^{(k)}, \lambda)  \\
		\bar{s}^{(k+1)} &  = \underset{\bar{s}}{\operatorname{argmin}} \;
		L_a(\bar{s}, x^{(k+1)}, \y^{(k+1)},\z^{(k+1)}, 
		\xh^{(k)},\yh^{(k)},\zh^{(k)},    \lambda)  \\
		\yh^{(k+1)}  &  = \yh^{(k)} +  c(\gamma\oplus({\bf 
			g}*(E^{(j)}\bar{s}^{(k+1)})) - \y^{(k+1)})   \\
		\zh^{(k+1)}  &  = \zh^{(k)} +  c((1- \gamma)\oplus({\bf 
			h}*(E^{(j)}\bar{s}^{(k+1)})) - \z^{(k+1)})   \\
		\xh^{(k+1)}  &  = \xh^{(k)} +  c(E^{(j)}\bar{s}^{(k+1)} - x^{(k+1)})   
	\end{align}
	The solution to the first and third minimizations are well-known from the 
	literature
	and they can be obtained in a single step  using exact formulas.  We show 
	that 
	the solution to the second problem can be obtained in a single step and we 
	derive
	the exact formula.  Finally, the fourth minimization problem is a quadratic 
	problem
	that can be solved iteratively using the conjugate gradient method.  We 
	will derive
	the computational expression for implementing the conjugate gradient 
	iterations.
	
	\subsubsection{First minimization: the $\y$-problem}
	
	We first rewrite the minimization problem by showing the specific terms in 
	$L_a(\bar{s}^{(k)}, x^{(k)}, \y,\z^{(k)}, \xh^{(k)},\yh^{(k)},\zh^{(k)}, 
	\lambda)$ that
	depend on $\y$:
	\begin{equation}
		\y^{(k+1)}   = \underset{\y}{\operatorname{argmin}} \;\;
		\lambda \|\y\|_2 + \frac{c}{2} \|\gamma\oplus({\bf 
			g}*(E^{(j)}\bar{s}^{(k)}))-\y\|_2^2 
		+  \innerprod{\hat{\y}}{\gamma\oplus({\bf g}*(E^{(j)}\bar{s}^{(k)}))-\y}
	\end{equation}
	The above problem is equivalent to the following problem, in which,  the 
	cost 
	differs
	only by a constant independent of $\y$:
	\begin{equation}
		\y^{(k+1)}   = \underset{\y}{\operatorname{argmin}} \;\; 
		\lambda \|\y\|_2     + \frac{c}{2} \|\bar{\y}^{(k)} -\y\|_2^2 
	\end{equation}
	where
	\begin{equation}
		\bar{\y}^{(k)}  = |\gamma\oplus({\bf g}*(E^{(j)}\bar{s}^{(k)}))  + 
		\frac{1}{c}\hat{\y}^{(k)}
	\end{equation}
	The above minimization problem separable across pixel indices. 
	Hence it can be written as a minimization problem in $\mathbb{R}^2$ as
	\begin{equation}
		\y^{(k+1)} (\r)   = \underset{\y(\r)}{\operatorname{argmin}} \;\;
		\lambda  \|\y(\r)\|_2 + \frac{c}{2} \|\bar{\y}^{(k)}(\r) -\y(\r)\|_2^2 
	\end{equation}
	The solution is well-known
	and is given by \cite{proximal_parikh}
	\begin{equation}
		\y^{(k+1)}(\r) = max(0, 
		\|\bar{\y}^{(k)}(\r)\|-\lambda/c)\bar{\y}^{(k)}(\r).
	\end{equation}
	
	\subsubsection{Second minimization: the $\z$-problem}
	\label{sec:zprob}

	In the cost of equation \eqref{eq:isocost}, if $\zeta$ has a value $0.5$ 
	everywhere
	independent of pixel location, then the $\z$-problem can be solved
	using the solution derived by Lefkimmiatis et al \cite{HS_13}.  However, 
	even if $\zeta$  does not satisfy this condition,  their solution can be
	extended by using the same techniques used by the them, which will be
	the main focus of this sub-section.
	
	Similar to how we handled   the first problem, we write the second problem 
	by showing the specific terms in 
	$L_a(\bar{s}^{(k)}, x^{(k)}, \y^{(k+1)},\z, \xh^{(k)},\yh^{(k)},\zh^{(k)}, 
	\lambda)$ that
	depend on $\z$:
	\begin{equation}
		\resizebox{\hsize}{!}{
			$\z^{(k+1)}   = \underset{\z}{\operatorname{argmin}} \;
			\lambda \langle \zeta,   |{\pmb \Lambda}_1({\bf z)|}\rangle +
			\lambda  \langle 1-\zeta,   |{\pmb \Lambda}_2({\bf z)|}\rangle
			+ \frac{c}{2} \|(1- \gamma)\oplus({\bf 
				h}*(E^{(j)}\bar{s}^{(k)}))-\z\|_2^2 
			+  \innerprod{\hat{\z}^{(k)}}{(1- \gamma)\oplus({\bf 
					h}*(E^{(j)}\bar{s}^{(k)}))-\z}$
		}
	\end{equation}
	The equivalent minimization problem that is in the format suitable
	for deriving the closed-form solution is given by
	\begin{equation}
		\z^{(k+1)}   = \underset{\z}{\operatorname{argmin}} \;\; 
		\lambda \langle \zeta,   |{\pmb \Lambda}_1({\bf z)|}\rangle 
		+ \lambda  \langle 1-\zeta,   |{\pmb \Lambda}_2({\bf z)|}\rangle
		+ \frac{c}{2} \|\bar{\z}^{(k)} -\z\|_2^2 
	\end{equation}
	where
	\begin{equation}
		\bar{\z}^{(k)}  = (1- \gamma)\oplus({\bf h}*(E^{(j)}\bar{s}^{(k)}))  + 
		\frac{1}{c}\hat{\z}^{(k)}
	\end{equation}
	Here too,  the above minimization problem is separable across pixel 
	indices. 
	Hence it can be written as a minimization problem in $\mathbb{R}^3$ as
	\begin{equation}
		\z^{(k+1)}(\r)   = \underset{\z(\r)}{\operatorname{argmin}} \;\; 
		\lambda  \zeta(\r) |{\pmb \Lambda}_1({\bf z(r)})|
		+ \lambda  (1-\zeta(\r)) |{\pmb \Lambda}_2({\bf z(r)})|
		+ \frac{c}{2} \|\bar{\z}^{(k)}(\r) -\z(\r)\|_2^2
	\end{equation}
	As mentioned before, the solution to the above problem is known for the 
	case when  $\zeta(\r) =0.5$ 
	\cite{proximal_parikh},
	and it is not known in general.  We give the solution in the following 
	proposition.
	
	\begin{prop}
		The solution to the above problem is given by
		$$\z^{(k+1)}(\r)  = \mathlarger{\cal E}\left(
		{\cal T}\big({\pmb \Lambda}_1(\bar{\bf z}^{(k+1)}(\r)), 
		\lambda\zeta(\r)/c\big),\;
		{\cal T}\big({\pmb \Lambda}_2(\bar{\bf 
			z}^{(k+1)}(\r)),\lambda(1-\zeta(\r))/c\big),\;
		\Theta(\bar{\bf z}^{(k+1)}(\r))\right),$$
		where ${\cal T}(x,y)$  denotes soft-thresholding applied on $x$ with 
		threshold $y$, where ${\cal E}(\cdot)$ is the operator  defined in 
		Section 
		\ref{sec:notcon}.
	\end{prop}
	\begin{proof}   
		Let
		$\eta_1={\pmb \Lambda}_1({\bf z}(\r))$,  $\eta_2={\pmb \Lambda}_2({\bf 
			z}(\r))$,
		$\bar{\eta}_1={\pmb \Lambda}_1(\bar{\bf z}^{(k+1)}(\r))$,  and 
		$\bar{\eta}_2={\pmb \Lambda}_2(\bar{\bf z}^{(k+1)}(\r))$,
		Further,  let $\theta = \Theta ({\bf z}(r))$,   $\bar{\theta} = \Theta 
		(\bar{\bf z}^{(k+1)}(r))$, where $\Theta(\cdot)$ is the operator 
		defined in
		Section \ref{sec:notcon}.
		Then the minimization problem can be written as
		\begin{equation}
			(\eta^{*}_1, \eta^{*}_2, \theta^*) = \underset{\eta_1, \eta_2, 
			\theta 
			}{\operatorname{argmin}} \;\; 
			c_1|\eta_1| + c_2 |\eta_2|  + \frac{1}{2} \|{\cal E}(\eta_1, 
			\eta_2, 
			\theta)-
			{\cal E}(\bar{\eta}_1, \bar{\eta}_2, \bar{\theta})\|_2^2,
		\end{equation}
		where $c_1 = \frac{1}{c}\lambda  \zeta(\r)$ and $c_2 = 
		\frac{1}{c}\lambda 
		(1- \zeta(\r))$. By following the idea
		of Lefkimmiatis et al \cite{HS_13},   we can conclude that the term  
		$\frac{1}{2} \|{\cal E}(\eta_1, \eta_2, \theta)-
		{\cal E}(\bar{\eta}_1, \bar{\eta}_2, \bar{\theta})\|_2^2$ attains the 
		minimum 
		with respect to $\theta$ at $\bar{\theta}$,
		i.e.,  $\theta^*=\bar{\theta}$.  At the minimum, 
		this part of the cost becomes equal to 
		$\frac{1}{2}(\eta_1-\bar{\eta}_1)^2 
		+ \frac{1}{2}(\eta_2-\bar{\eta}_2)^2$.
		Hence the simplified minimization problem becomes the following:
		\begin{equation}
			(\eta^{*}_1, \eta^{*}_2) = \underset{\eta_1, \eta_2 
			}{\operatorname{argmin}} \;\; 
			c_1|\eta_1| + c_2 |\eta_2|  + \frac{1}{2}(\eta_1-\bar{\eta}_1)^2 + 
			\frac{1}{2}(\eta_2-\bar{\eta}_2)^2,
		\end{equation}
		The solution is well-known and is given by \cite{proximal_parikh}    
		$\eta^{*}_1 = {\cal T}\big(\bar{\eta}_1, c_1\big)$  and 
		$\eta^{*}_2 = {\cal T}\big(\bar{\eta}_2, c_2\big)$.  With this,  the 
		required solution can be expressed
		as
		$\z^{(k+1)}(\r) = {\cal E}(n_1^*,n_2^*,\theta^*)= 
		{\cal E}( {\cal T}\big(\bar{\eta}_1, c_1\big), {\cal 
		T}\big(\bar{\eta}_1, 
		c_1\big),\bar{\theta})$.
		Substituting back for $\bar{\eta}_1$, $\bar{\eta}_2$, $\bar{\theta}$, 
		$c_1$ 
		and $c_2$  gives
		final expression.
	\end{proof}
	
	\subsubsection{Third minimization: the $x$-problem}
	
	As before,   we write the second problem  by showing the specific terms in 
	$L_a(\bar{s}^{(k)}, x, \y^{(k+1)},\z^{(k+1)}, 
	\xh^{(k)},\yh^{(k)},\zh^{(k)}, 
	\lambda)$ that
	depend on $x$:
	\begin{equation}
		x^{(k+1)}   = \underset{x}{\operatorname{argmin}} \;
		{\cal B}(x)  + \frac{c}{2} \|E^{(j)}\bar{s}^{(k)}-x\|_2^2 
		+  \innerprod{\hat{x}}{E^{(j)}\bar{s}^{(k)}-x}
	\end{equation}
	An equivalent problem is
	\begin{equation}
		x^{(k+1)}   = \underset{x}{\operatorname{argmin}} \;
		{\cal B}(x)  + \frac{c}{2} \|\bar{x}^{(k)}-x\|_2^2,
	\end{equation}
	where
	\begin{equation}
		\bar{x}^{(k)} = E^{(j)}\bar{s}^{(k)} + \frac{1}{c} \bar{x}^{(k)}
	\end{equation}
	Again, this problem is separable across pixels,  and can be written as
	\begin{equation}
		x^{(k+1)}(\r)   = \underset{x(\r)}{\operatorname{argmin}} \;
		{\cal B}(x(\r))  + \frac{c}{2} \|\bar{x}^{(k)}(\r)-x(\r)\|_2^2 
	\end{equation}
	The solution to this  problem is well known and it is given
	by $x^{(k+1)}(\r)  = {\cal P}(\bar{x}^{(k)}(\r))$, where ${\cal P}(\cdot)$
	denotes the clipping its argument to the range within which 
	${\cal B}(x(\r))$ zero.
	
	\subsubsection{Fourth  minimization: the  $s$-problem}
	
	The cost
	$L_a(\bar{s}, x^{(k+1)}, \y^{(k+1)},\z^{(k+1)}, 
	\xh^{(k)},\yh^{(k)},\zh^{(k)}, 
	\lambda)$
	with respect to $\bar{s}$ is quadratic,  and the minimum is obtained by 
	equating
	the gradient of $L_a(\bar{s}, x^{(k+1)}, \y^{(k+1)},\z^{(k+1)}, 
	\xh^{(k)},\yh^{(k)},\zh^{(k)}, \lambda)$
	with respect $s$ to zero.  The equation is given by
	\begin{equation}
		\begin{split}
			& E^{T{(j)}}\left(
			\g^T*\left(\gamma^2\oplus({\bf g}*(E^{(j)}\bar{s}))\right)  
			\;\;+\;\;
			\h^T*\left((1- \gamma)^2\oplus({\bf h}*(E^{(j)}\bar{s}))\right)
			\;\; + \;\;  b^T*b*(E^{(j)}\bar{s}))  \;\;+ \;\; 
			E^{(j)}\bar{s}
			\right) \\
			& \;\;\;\;\;\;\; =  
			E^{T{(j)}}\left(
			\g^T*\left(\gamma^2\oplus\tilde{\y}^{(k+1)} \right)  \;\;+\;\;
			\h^T*\left((1- \gamma)^2\oplus\tilde{\z}^{(k+1)}\right)
			\;\; + \;\; \tilde{x}^{(k+1)}  \;\;+ \;\; b^T*m
			\right),
		\end{split}
	\end{equation}
	where
	\begin{align}
		\tilde{\y}^{(k+1)} & = \left(  \y^{(k+1)}-\frac{1}{c}\hat{\y}^{(k+1)}  
		\right) \\
		\tilde{\z}^{(k+1)} & = \left(  \z^{(k+1)}-\frac{1}{c}\hat{\z}^{(k+1)}  
		\right) \\
		\tilde{x}^{(k+1)} & = \left(  x^{(k+1)}-\frac{1}{c}\hat{x}^{(k+1)}  
		\right).
	\end{align}
	Here $m$ denotes the  inverse DFT of the 
	image obtained by inserting measured samples in a 2D array of zeros.
	Further,  $b$ denotes inverse DFT of  the binary image containing  ones
	at locations where Fourier samples were measured and zeros at other 
	locations.  The above equation can be solved by the method of
	conjugate gradients. 
	
	\section{Computation of spatially adaptive weights}
	\label{sec:wtcomp}
	
	Here we consider the problem of determining the adaptive weights with the
	guide image $\hat{s}(\r)$ given, which is represented by
	\begin{equation}
		(\hat{\beta}_1,  \hat{\beta}_2, \hat{\beta}_3) =\mathlarger{{\cal 
		M}_w}\Big(
		\scalebox{.7}[1.0]{\bf AdaptGuide}=\hat{s},   \tau \Big)
	\end{equation}
	Let $Y_{1}({\bf r},\hat{s}) =  \|({\bf g}*\hat{s})({\bf r}) \|_2$,
	$Y_{2}({\bf r},\hat{s}) =   |{\pmb \Lambda}_{1}(({\bf h}*\hat{s})({\bf 
			r}))|$,  
		and
		$Y_{3}({\bf r},\hat{s}) =   |{\pmb \Lambda}_{2}(({\bf h}*\hat{s})({\bf 
		r}))|$,
	Now the minimization problem can be written as
	\begin{equation}
		\label{eq:betaprob}
		\{\hat{\beta}_1, \hat{\beta}_2,\hat{\beta}_3\} =  
		\underset{{\beta}_1, {\beta}_2,{\beta}_3}
		{\operatorname{argmin}}
		\sum_{i=1}^3\beta_i({\bf r}) Y_{i}({\bf r},\hat{s}) -
		\tau \sum_{\bf r} \log(\beta_1({\bf r})\beta_2({\bf r})\beta_3({\bf 
		r})) 
	\end{equation}  
	subject to   $$0\le {\beta}_i(\r) \le 1 \;\; \mbox{and} 
	\sum_{i=1}^3{\beta}_i({\bf r}) = 1.$$
	The solution to this problem is given in the following proposition.
	\begin{prop}
		The solution is given by   $\hat{\beta}_i(\r) = \frac{\tau}{Y_{i}({\bf 
				r},\hat{s}) + \varphi({\bf r})}$  
		where $\varphi({\bf r})$ is such that   
		$0\le \hat{\beta}_i(\r) \;\; \mbox{and} \sum_{i=1}^3\hat{\beta}_i({\bf 
		r}) 
		= 1$.
	\end{prop}
	\begin{proof} 
		Since the minimization problem is independent of the pixel index
		${\bf r}$, we simplify the notation and  write the cost to be
		minimized as  $A(\beta_1, \beta_2, \beta_3) =
		\sum_{i=1}^n \beta_iy_i - \tau \log(\beta_1\beta_2\beta_3)$.  As a 
		first 
		step towards obtaining
		the solution, we  ignore the bound constraint and  apply Langrange 
		multiplier theory to 
		handle the equality constraint \cite{Bertsekas_NLP}.  The Lagrangian is 
		given 
		by 
		$A_l(\beta_1, \beta_2, \beta_3) =
		\sum_{i=1}^n \beta_iy_i - \tau \log(\beta_1\beta_2\beta_3) + \varphi
		(\sum_{i=1}^n \beta_i - 1)$, where $\varphi$ is the Lagrange's 
		multiplier.
		The minimum of the constrained optimization problem can be
		represented by the equations 
		$\frac{\partial}{\partial \beta_i}A_l(\beta_1, \beta_2, \beta_3)=0, 
		i=1,2,3$. 
		Representing the solution by $\hat{\beta}_i$ we get  
		$\hat{\beta}_i = \frac{\tau}{y_{i} + \varphi}$.  Here  $\varphi$  has 
		to be 
		determined
		such that the constraints
		$0\le {\beta}_i(\r) \le 1 \;\; \mbox{and} \sum_{i=1}^3{\beta}_i({\bf 
		r}) = 
		1$
		are satisfied.  Bringing back the dependence on pixel index ${\bf r}$  
		and 
		the 
		dependence of $Y_{i}$ on $\hat{s}$ gives the required expression for 
		$\hat{\beta}_i(\r)$.
	\end{proof}
	
	To determine  $\varphi({\bf r})$ satisfying {\bf Proposition 2},   
	we note that the sum $\sum_{i=1}^3{\beta}_i({\bf r})$ is a monotonically
	decreasing function of $\varphi({\bf r})$.   Hence $\varphi({\bf r})$
	that satisfies  $\sum_{i=1}^3{\beta}_i({\bf r})=1$  can be determined
	by binary search. To determine interval for binary search, 
	first we note that
	by imposing $0\le \hat{\beta}_i(\r), i=1,2,3$, we  have a lower bound on 
	$\varphi({\bf r})$
	that $\varphi({\bf r}) > -{\beta}_m(\r)$, where   ${\beta}_m(\r) = min  
	\{\hat{\beta}_i(\r), i=1,2,3\}$. To determine a more useful bound, we adopt 
	the 
	following strategy.  Lower bound on $\varphi({\bf r})$
	that ensures that $\sum_{i=1}^3\beta_i(\mathbf{r})>1$
	can be obtained by imposing $max(\beta_i(\mathbf{r}), i=1,2,3) >1$, which 
	leads 
	to
	the lower bound on $\varphi({\bf r})$  as  $\tau - Y_{m}({\bf r},\hat{s})$
	where  $Y_{m}({\bf r},\hat{s})=min(Y_{i}({\bf r},\hat{s}), i=1,2,3)$.
	upper bound on $\varphi({\bf r})$
	that ensures that $\sum_{i=1}^3\beta_i(\mathbf{r}) < 1$
	can be obtained by imposing $max(\beta_i(\mathbf{r}), i=1,2,3)  < 1/3$
	which leads to upper bound on 
	$\varphi({\bf r})$  as  $3\tau - Y_{m}({\bf r},\hat{s})$. 
	Using these bounds, a  binary search on $\varphi({\bf r})$ can be done
	such that $\sum_{i=1}^3{\beta}_i({\bf r})=1$.
	Since binary search 
	can 
	converge very fast,  we can assume that the minimization problem of equation
	\eqref{eq:betaprob} can be solved exactly, in the sense that convergence 
	can be
	achieved with a tolerance that is comparable with machine precision within a
	reasonable amount of time.

	\bibliographystyle{unsrt}  
	\bibliography{hcorosa}

\end{document}